\documentclass[a4paper,twocolumn,10pt,unpublished]{quantumarticle}
\pdfoutput=1
\usepackage[utf8]{inputenc}
\usepackage[english]{babel}
\usepackage[T1]{fontenc}
\usepackage{amsmath,amssymb,amsfonts,amsthm,thm-restate}
\usepackage[colorlinks=false,hidelinks]{hyperref}
\usepackage{cleveref}

\usepackage{tikz}
\usepackage{lipsum}
\usepackage{braket}
\usepackage{relsize}
\usepackage{tikz}
\usepackage{caption}
\usepackage{subcaption}
\usetikzlibrary{quantikz2}
\usepackage{mdframed}
\usepackage{enumitem}

\newtheorem{theorem}{Theorem}
\newtheorem{lemma}{Lemma}

\DeclareMathOperator{\tr}{Tr}
\DeclareMathOperator*{\argmax}{arg\,max}
\DeclareMathOperator*{\locc}{LOCC}
\DeclareMathOperator*{\cptp}{CPTP}
\DeclareMathOperator*{\sign}{sign}
\newcommand{\ketbra}[2]{|#1\rangle\langle#2|}
\newcommand{\ketbras}[1]{\ketbra{#1}{#1}}

\begin{document}

\title{Simulating Quantum State Transfer between Distributed Devices using Noisy Interconnects}
\author{Marvin Bechtold}
\affiliation{Institute of Architecture of Application Systems, University of Stuttgart, Universitätsstraße 38, 70569 Stuttgart, Germany}
\orcid{0000-0002-7770-7296}
\email{bechtold@iaas.uni-stuttgart.de}
\author{Johanna Barzen}
\email{barzen@iaas.uni-stuttgart.de}
\affiliation{Institute of Architecture of Application Systems, University of Stuttgart, Universitätsstraße 38, 70569 Stuttgart, Germany}
\orcid{0000-0001-8397-7973}
\author{Frank Leymann}
\affiliation{Institute of Architecture of Application Systems, University of Stuttgart, Universitätsstraße 38, 70569 Stuttgart, Germany}
\orcid{0000-0002-9123-259X}
\email{leymann@iaas.uni-stuttgart.de}
\author{Alexander Mandl}
\affiliation{Institute of Architecture of Application Systems, University of Stuttgart, Universitätsstraße 38, 70569 Stuttgart, Germany}
\orcid{0000-0003-4502-6119}
\email{mandl@iaas.uni-stuttgart.de}
\maketitle

\begin{abstract}
    Scaling beyond individual quantum devices via distributed quantum computing relies critically on high-fidelity quantum state transfers between devices, yet the quantum interconnects needed for this are currently unavailable or expected to be significantly noisy.
These limitations can be bypassed by simulating ideal state transfer using quasiprobability decompositions~(QPDs).
Wire cutting, for instance, allows this even without quantum interconnects.
Nevertheless, QPD methods face drawbacks, requiring sampling from multiple circuit variants and incurring substantial sampling overhead.
While prior theoretical work showed that incorporating noisy interconnects within QPD protocols could reduce sampling overhead relative to interconnect quality, a practical implementation for realistic conditions was lacking.
Addressing this gap, this work presents a generalized and practical QPD for state transfer simulation using noisy interconnects to reduce sampling overhead.
The QPD incorporates a single tunable parameter for straightforward calibration to any utilized interconnect.
To lower practical costs, the work also explores reducing the number of distinct circuit variants required by the QPD.
Experimental validation on contemporary quantum devices confirms the proposed QPD’s practical feasibility and expected sampling overhead reduction under realistic noise.
Notably, the results show higher effective state transfer fidelity than direct transfer over the underlying noisy interconnect. \end{abstract}

\section{Introduction}
While quantum computing holds immense theoretical potential to outperform classical computers on specific tasks~\cite{Shor1997,Liu2021a}, its practical realization faces fundamental scalability challenges.
Current quantum devices contain hundreds of noisy physical qubits~\cite{AbuGhanem2025,MontanezBarrera2025}, while practical applications require many more fault-tolerant qubits~\cite{Acharya2024}. 
As scaling individual quantum devices remains technically challenging~\cite{Mohseni2024}, distributed quantum computing emerges as a promising and complementary scaling strategy, aiming to connect multiple smaller devices via quantum interconnects and classical communication~\cite{Barral2024}.
This distributed architecture relies crucially on high-fidelity quantum state transfer between devices~\cite{Caleffi2024}. 
However, establishing high-fidelity state transfers is challenging, as quantum interconnects, whether initial short-range versions between adjacent chips~\cite{Gold2021,Conner2021} or future long-range connections via cables~\cite{Zhong2019,Main2025}, are expected to suffer from significant noise and errors.

To overcome limitations from missing or noisy quantum interconnects, the action of an ideal quantum state transfer can be simulated using a \emph{quasiprobability decomposition} (QPD)~\cite{Temme2017,Piveteau2024,Brenner2023}. 
Such a QPD represents the ideal transfer in terms of the inherently noisy operations physically available between the distributed devices.
The simulation reconstructs the noiseless transfer by sampling these available operations according to the QPD. 
Classical post-processing then combines the results from the sampled operations to recover the target measurement statistics corresponding to the ideal transfer.
When no quantum interconnects are available and the QPD's operations are restricted to local quantum operations on individual devices possibly coordinated via classical communication, this method is known as \emph{wire cutting}~\cite{Bechtold2023b}.
However, key limitations when using QPDs include: (i)~the need to transpile and execute multiple distinct circuit variants~\cite{Harada2024}, and (ii)~the total number of measurements required to achieve a fixed accuracy, known as the \emph{sampling overhead}, typically scales exponentially with the number of qubits in the simulated transfer~\cite{Brenner2023}.
Recent advances propose mitigating this sampling overhead by leveraging noisy quantum interconnects to generate shared entanglement between devices~\cite{Bechtold2024a,Bechtold2025}. 
Because the interconnects themselves introduce noise, the generated entanglement is imperfect, typically yielding shared states that are not only \emph{non-maximally entangled}~(NME) but also mixed (impure).
Incorporating such shared NME states into wire cutting protocols can theoretically reduce the sampling overhead, with more significant improvements for states containing higher entanglement.

However, these prior protocols rely on the assumption of idealized, pure NME states~\cite{Bechtold2024a,Bechtold2025}. 
This assumption deviates significantly from the reality of noisy interconnects, which inevitably produce mixed NME states~\cite{Zhong2019,Liu2021b,Tao2024}, thereby limiting the practical applicability of existing methods.
Our work addresses this critical gap and advances these techniques through three key contributions.
First, we develop a generalized QPD for quantum state transfer designed for feasible execution on near-term distributed quantum devices that reduces sampling overhead by incorporating noisy quantum interconnects. 
This QPD can either directly utilize the noisy interconnects or employ the mixed NME states they generate as a resource, extending previous work by removing the assumption of pure NME states.
Crucially, our QPD formulation features a single tunable parameter, simplifying calibration: it can be readily adapted to any specific noisy interconnect using only minimal experimental characterization of the channel quality using a method detailed in this work.
Second, to lower the transpilation cost and the number of distinct circuits to execute, we investigate strategies to reduce the number of distinct circuit variants required by the QPD and analyze the trade-offs involving implementation errors introduced by these simplification techniques.
Third, we experimentally validate the feasibility of our QPD-based simulation of the state transfer on contemporary quantum devices. 
This includes calibrating the QPD based on noisy interconnect characteristics and subsequently executing the simulation.
Our experimental results confirm a reduction in sampling overhead under realistic noise, thereby bridging the gap between prior theoretical decompositions and practical implementation on noisy quantum devices.
Notably, these experiments reveal that employing noisy interconnects within our QPD-based simulation can achieve a higher effective state transfer fidelity compared to attempting direct state transfer over the same physical interconnects, even when the QPD incorporates approximations designed to reduce the number of circuit variants.

The structure of this paper is as follows: 
\Cref{sec:preliminaries} introduces the fundamental concepts and background.
\Cref{sec:main} presents our novel QPD for simulating quantum state transfer using noisy interconnects.
Subsequently, \Cref{sec:experiments} details our experimental methodology and results, encompassing both simulations and implementations on real quantum devices.
\Cref{sec:discussion} discusses our findings and limitations. 
To place our results in context, \Cref{sec:related_work} provides an overview of related work, before \Cref{sec:conclusion} concludes the work.
 \section{Background}\label{sec:preliminaries}
This section introduces the concepts underpinning the simulation of quantum state transfer via QPDs.
We start with introducing the used notation for quantum systems and states, as well as Pauli operators and their simultaneous measurement.
Building on this, we describe the modeling of state transfers, including imperfections, using quantum channels. 
We then examine their implementation using local operations and classical communication with shared entangled states as resources.
Following this, we introduce channel twirling, a technique for simplifying channel structure, which is subsequently applied to reduce the complexity of errors in the noisy physical channel.
We then detail QPDs themselves and explain how they enable the simulation of ideal transfers using physically constrained and imperfect operations.

\subsection{Quantum systems, states, and entanglement}
A quantum system is mathematically described by a Hilbert space $A$, with bounded linear operators $L(A)$ acting on it.
The physical state space of the system consists of density operators $D(A) \subset L(A)$, defined as positive semidefinite Hermitian operators with unit trace~\cite{Wilde2017}.
We denote states as $\rho_A \in D(A)$, where subscripts explicitly indicate the associated system when necessary.
Pure states $\ket{\psi} \in A$ are represented as rank-$1$ projectors $\psi = \ketbras{\psi} \in D(A)$.

When considering distributed systems, such as two devices $A$ and $B$, their joint state is described by a density operator $\rho_{AB} \in D(A\otimes B)$ on the composite Hilbert space $A \otimes B$. 
A state $\rho_{AB}$ is called separable if it can be written as a probabilistic mixture of product states: $\rho_{AB} = \sum_i p_i \rho_A^{(i)} \otimes \rho_B^{(i)}$, where $\{p_i\}_i$ is a probability distribution and $\rho_X^{(i)} \in D(X)$. 
The set of all separable states between $A$ and $B$ is denoted $S(A,B)\subset D(A\otimes B)$. 
States that are not separable, i.e., $\rho_{AB} \notin S(A,B)$, are entangled. 
A key example of an entangled state is the $2n$-qubit maximally entangled state~\cite{Wilde2017}, which is given in the computational basis $\{\ket{\vec{k}}\}_{\vec{k}\in\{0,1\}^n}$ by
\begin{align}
 \ket{\Phi_n}_{AB} = \frac{1}{\sqrt{2^n}}\sum_{\vec{k}\in\{0,1\}^n} \ket{\vec{k}}_A\ket{\vec{k}}_B.
\end{align}
We denote its density operator as $\Phi_{AB}$ with size $n$ implicit in systems $A$ and $B$, or $\Phi_n$ if the factorization in two systems is clear from the context or irrelevant. 

\subsection{Pauli operators}
Pauli operators play a foundational role in quantum computation and information, forming a basis for operators $L(A)$ that is essential for describing and analyzing quantum systems, processes, and measurements.
The single-qubit Pauli operators are denoted by $I$, $X$, $Y$, and $Z$~\cite{Wilde2017}.
Their $n$-qubit generalizations can be specified using binary vectors $\vec{a}=(a_0, \ldots, a_{n-1}) \in \{0,1\}^{n}$ as
\begin{align}\label{eq:X_and_Z}
 X_{\vec{a}} := \bigotimes_{i=0}^{n-1} X^{a_i},\quad Z_{\vec{a}} := \bigotimes_{i=0}^{n-1} Z^{a_i},
\end{align}
with $P^{a_i} = I$ if $a_i=0$ and $P^{a_i} = P$ if $a_i=1$, for $P \in \{X, Z\}$.

A general Pauli operator $P_{\vec{a}}$ combines $X$- and $Z$-components:
\begin{align}
 P_{\vec{a}} = X_{\vec{x}}Z_{\vec{z}}
\end{align}
where $\vec{x},\vec{z} \in \{0,1\}^n$, and the $2n$-dimensional vector $\vec{a} = (\vec{x},\vec{z}) \in \{0,1\}^{2n}$ concatenates them.
For convenience, we sometimes represent $\vec{a}$ by its integer equivalent $a = \sum_{i=0}^{2n-1} a_i 2^i$, writing $P_{a}=P_{\vec{a}}$, and similarly $X_x = X_{\vec{x}}$ and $Z_z = Z_{\vec{z}}$.
A detailed analysis of the properties of the Pauli operators required for our results can be found in \Cref{sec:appendix_paulis}.

The set of all such operators, including phase factors $\{\pm 1, \pm i\}$, forms the $n$-qubit Pauli group $\mathcal{P}_n$:
\begin{align}
 \mathcal{P}_n := \left\{i^kP_{\vec{a}}\middle| k\in\{0,1,2,3\}, \, \vec{a}\in \{0,1\}^{2n}\right\}.
\end{align}
As global phases are physically unobservable, we often work with the quotient group $\mathcal{Q}_n := \mathcal{P}_n / \{\pm 1, \pm i\}$, and denote as $\mathcal{Q}_n^* := \mathcal{Q}_n \setminus {I^{\otimes n}}$ the set of non-identity operators in this group.
The projection $\pi\colon \mathcal{P}_n \to \mathcal{Q}_n$, defined by
\begin{align}\label{eq:projection_pauli}
\pi(\alpha P) &= P \quad \text{for } \alpha \in \{\pm 1, \pm i\},\ P \in \mathcal{Q}_n,
\end{align}
identifies operators up to phase equivalence.

The Pauli operators in $\mathcal{Q}_n$ form an orthogonal basis for the space of linear operators $L(A)$ for an $n$-qubit Hilbert space $A$~\cite{Lawrence2002}, under the Hilbert-Schmidt inner product:
\begin{align}\label{eq:pauli_orthogonality}
\tr[P_a P_b] &= 2^n \delta_{a,b} \quad \forall P_a, P_b \in \mathcal{Q}_n.
\end{align}

This completeness as an operator basis means that any observable relevant to quantum measurements can be expressed as a linear combination of these Pauli operators.
A key property arising from this representation, which enables efficient measurements and will be utilized later, is that commuting Pauli operators can be measured simultaneously.
If $P_{a}, P_{b} \in \mathcal{Q}_n $ commute ($P_{a}P_{b} = P_{b}P_{a}$), they share a common eigenbasis, allowing their measurement outcomes to be determined from a single experiment involving a specific basis rotation~\cite{Gokhale2020}.

The set $\mathcal{Q}_n^*$ can be partitioned into $2^n+1$ disjoint subsets $S_j$~\cite{Lawrence2002}. 
Each of these subsets is formed by $2^n-1$ operators that all pairwise commute, and they are maximal in this respect: no additional operator from $\mathcal{Q}_n^*$ can be included in an $S_j$ without disrupting this complete pairwise commutativity. 
Such subsets are termed maximally commuting. 
This partitioning into maximally commuting subsets provides the largest sets of simultaneously measurable Pauli operators, crucial for minimizing measurement circuits in protocols like the QPD discussed later.

For each maximally commuting subset $S_j$, there exists a unitary transformation $V_j$ that simultaneously diagonalizes all its elements, yielding the representation~\cite[Lemma 3]{Harada2024}:
\begin{align}\label{eq:S_j}
S_j = \left\{s_{j,\vec{a}}V_jZ_{\vec{a}}V^{\dagger}_j \mid \vec{a}\in \{0,1\}^n, \vec{a} \neq \vec{0}\right\},
\end{align}
where operators $Z_{\vec{a}}$ are diagonal in the computational basis and $s_{j,\vec{a}} \in \{-1,1\}$ are sign factors.

Given that the QPD presented in \Cref{sec:main} uses these $V_j$ transformations, their efficient implementation on a quantum computer is vital.
Indeed, for $n$ qubits, each required unitary operator $V_j$ can be implemented efficiently as a quantum circuit with a depth of at most $n+2$, using $n$ Hadamard gates, at most $n$ Phase gates, and $\frac{n(n-1)}{2}$ controlled-$Z$ gates~\cite[Lemma 4]{Harada2024}. 
In the single-qubit case, the three necessary operators are:
\begin{align}\label{eq:single_qubit_mub_transformations}
V_0 = I, \quad V_1 = H, \quad V_2 = SH,
\end{align}
where $H$ is the Hadamard gate and $S$ is the phase gate~\cite{Harada2024}.

\subsection{State transfer via quantum channels}\label{sec:channels}
Independent of the specific protocol, any quantum state transfer between devices, including the inevitable noise affecting the process, can be described using \emph{quantum channels}~\cite{Wilde2017}.
A quantum channel $\mathcal{C}: L(A) \rightarrow L(B)$ is a linear map that transforms states of system $A$ into valid states of system $B$, i.e., $\mathcal{C}(\rho_A) \in D(B)$ for all $\rho_A \in D(A)$.
To ensure physical validity, these maps must be \emph{completely positive trace-preserving~(CPTP)}~\cite{Wilde2017}.
The set of such maps from $A$ to $B$ is denoted as $\cptp(A, B)$, with $\cptp(A) = \cptp(A, A)$ for maps acting on a single system. The ideal, noiseless transfer of an $n$-qubit state from $A$ to $B$ is represented by the identity channel $\mathcal{I} \in \cptp(A,B)$, where $\mathcal{I}(\rho_A) = \rho_B$.

To quantify the difference between quantum channels, e.g., comparing an implemented channel $\mathcal{C}$ to an ideal state transfer $\mathcal{I}$, one analyzes their output states.
A measure of distinguishability between any two quantum states $\sigma_1$ and $\sigma_2$, is their trace distance $\|\sigma_1 - \sigma_2\|_{1}$, derived from the trace norm $\|X\|_1 = \tr\left[\sqrt{X^{\dagger}X}\right]$~\cite{Wilde2017}.
This distance is bounded such that $0\le\|\sigma_1 - \sigma_2\|_{1}\le2$~\cite{Fuchs1999}.
To find the difference between channels $\mathcal{C}_1$ and $\mathcal{C}_2$, this trace distance of their outputs is maximized over all possible inputs.
This requires considering all input states $\rho \in D(R\otimes A)$, crucially allowing the primary system $A$ to be initially entangled with an ancillary system $R$.
This yields the \emph{diamond norm distance} $\|\mathcal{C}_1 - \mathcal{C}_2\|_{\diamond}$~\cite{Wilde2017}.
It is the standard operational metric as it captures worst-case channel distinguishability, calculated by:
\begin{align}\label{eq:diamond_norm_def}
\begin{split}
 &\|\mathcal{C}_1 - \mathcal{C}_2\|_\diamond \\&= \max_{\rho \in D(R\otimes A)} \| (\mathcal{I}_R \otimes \mathcal{C}_1)(\rho) - (\mathcal{I}_R \otimes \mathcal{C}_2)(\rho) \|_1,
\end{split}
\end{align}
The value of the diamond norm is bounded as
\begin{align}\label{eq:bounds_diamond_norm}
    0 \le \|\mathcal{C}_1 - \mathcal{C}_2\|_\diamond \le 2,
\end{align}
where a value of $0$ indicates that the channels are identical, while the maximum value of $2$ signifies that they are perfectly distinguishable by some measurement.
These bounds are inherited directly from the trace distance of the output states.
Consequently, for any channel $\mathcal{C}$, its deviation from the ideal state transfer is quantified by $\|\mathcal{C} - \mathcal{I}\|_{\diamond}$.

While the diamond norm quantifies the worst-case distinguishability between quantum channels, the \emph{entanglement fidelity} $F$ measures an average-case distance of a channel to the identity channel $\mathcal{I}$~\cite{Haah2023}.
Specifically, for an $n$-qubit channel $\mathcal{C}\in \cptp(A, B)$ with $\dim(A)=\dim(B)=2^n$, its entanglement fidelity quantifies how well $\mathcal{C}$ preserves the maximally entangled state $\Phi_{RA}$, initially shared between system $A$ and an $n$-qubit reference system $R$~\cite{Nielsen2002}.
It is defined as
\begin{align}\label{eq:entanglement_fidelity}
 F(\mathcal{C}) = \braket{\Phi_{RB}|(\mathcal{I}_R\otimes \mathcal{C})(\Phi_{RA})|\Phi_{RB}}.
\end{align}
A value of $F(\mathcal{C})=1$ indicates that the channel perfectly preserves this entanglement structure, implying $\mathcal{C}=\mathcal{I}$.

Quantum channels can be described using various mathematical representations. 
For the purposes of this work, the $\chi$-matrix formalism is particularly beneficial, as it expresses an $n$-qubit channel $\mathcal{C}$ in terms of the Pauli operator basis $\mathcal{Q}_n$~\cite{Wood2015}. 
In this representation, the action of the channel on an input state $\rho$ is given by:
\begin{align}
 \mathcal{C}(\rho) &= \sum_{a,b =0}^{2^{2n}-1}\chi_{ab}P_a \rho P_b.
\end{align}
A $\chi$-matrix represents a valid quantum channel, i.e., a linear CPTP map, if it is positive-semidefinite and Hermitian, and additionally satisfies the condition $\sum_{a,b}\chi_{ab} P_a P_b = I^{\otimes n}$~\cite{Wood2015}.
The entanglement fidelity of the channel is given in this representaion by~\cite{Wood2015}
\begin{align}\label{eq:ent_fidelity_chi}
    F(\mathcal{C}) = \chi_{00}.
\end{align}

The $\chi$-matrix partitions the channel’s action into two physically distinct components:
\begin{align}\label{eq:coherent_incoherent_part}
 \mathcal{C}(\rho) &= \underbrace{\sum_{a}\chi_{aa}P_a \rho P_a}_{\text{incoherent part}} + \underbrace{\sum_{a\ne b}\chi_{ab}P_a \rho P_b}_{\text{coherent part}}.
\end{align}
The diagonal terms describe the incoherent part, representing probabilistic application of Pauli errors with $\chi_{aa} \ge 0$ and $\sum_a \chi_{aa} = 1$.
The off-diagonal terms, i.e., $\chi_{ab}$ for $a \ne b$, describe the coherent part, capturing quantum interference effects between different Pauli operations. 
A channel is called a \emph{Pauli channel} if its $\chi$-matrix is diagonal, i.e., $\chi_{ab}=0$ for $a \neq b$, meaning it only consists of probabilistic Pauli errors:
\begin{align}\label{eq:pauli_channel}
 \mathcal{C}_{\text{Pauli}}(\rho) = \sum_{a=0}^{2^{2n}-1}\chi_{aa}P_a \rho P_a.
\end{align}

A fundamental Pauli channel is the \emph{depolarizing channel} $\mathcal{D}_p$~\cite{Quek2024}, defined as:
\begin{align}\label{eq:depol_channel}
 \mathcal{D}_p(\varphi) := p \mathcal{I}^{\otimes n}(\varphi) + \frac{1-p}{2^{2n}-1}\sum_{P \in \mathcal{Q}_n^*} P \varphi P.
\end{align}
This channel leaves the state $\varphi$ unchanged with probability $p$ and applies a uniformly random non-identity Pauli error with probability $1-p$. 
Its entanglement fidelity is $F(\mathcal{D}_p) = p$ (see \Cref{eq:ent_fidelity_chi}).
Using \Cref{eq:depol_channel}, we can represent the depolarizing channel $\mathcal{D}_p$ as 
\begin{align}\label{eq:depol_channel_with_D0}
    \mathcal{D}_p = p\mathcal{I}^{\otimes n} + (1-p)\mathcal{D}_0.
\end{align}

\subsection{State transfer using local operations and classical communication}
For distributed quantum devices $A$ and $B$, which are typically spatially separated, operational capabilities are generally restricted to actions performed locally on each individual subsystem, possibly coordinated by classical communication. 
This class of operations, motivated by the physical constraints inherent in such distributed scenarios, is known as \emph{local operations and classical communication}~(LOCC), denoted $\locc(A,B) \subset \cptp(A\otimes B)$~\cite{Chitambar2014}.

Using LOCC, quantum teleportation implements a state transfer between systems $A$ and $B$ by leveraging the shared maximally entangled state $\Phi_{AB}$ as its resource~\cite{Bennett1993}. 
Ideally, with a perfect resource $\Phi_{AB}$, this process perfectly transfers an arbitrary quantum state from sender to receiver, realizing an identity channel for the transferred state.
The protocol for teleporting a single-qubit state $\varphi$, as illustrated in \Cref{fig:teleportation_circuit}, can be extended to transfer an $n$-qubit state~\cite{Bechtold2025}. 
Such an extension involves employing $n$ parallel instances of this single-qubit procedure, which requires $n$ individual shared, maximally entangled states, collectively denoted as $\Phi_n$.

\begin{figure}
 \centering
\includegraphics[trim={0 -0.4cm 0 -0.1cm},clip]{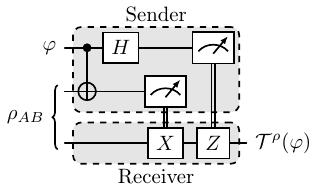}
 \caption{Single-qubit teleportation protocol with resource state $\rho_{AB}$.}
 \label{fig:teleportation_circuit}
\end{figure}

In practice, the shared resource state $\rho_{AB}$ may be imperfect, i.e., an NME state, rather than the ideal $\Phi_n$.
Using such an NME state for teleportation introduces errors, transforming the ideal identity channel into a noisy Pauli channel $\mathcal{T}^{\rho}$~\cite{Bowen2001,Gu2004}:
\begin{align}
 \mathcal{T}^{\rho}: \varphi \mapsto \sum_{a = 0}^{2^{2n-1}}\braket{\Phi^{a}|\rho|\Phi^{a}}P_a\varphi P_a.
\end{align}
Here, $\ket{\Phi^{a}} = (P_a \otimes I^{\otimes n}) \ket{\Phi_n}$ is the image of the maximally entangled state after applying the corresponding Pauli operator $P_a \in \mathcal{Q}_n$.
The error probability associated with each Pauli operator $P_a$ is determined by the fidelity $\braket{\Phi^a | \rho | \Phi^a}$. 
This establishes a direct connection between the quality of the resource state and the noise profile of the resulting channel.
The entanglement fidelity of the teleportation channel $\mathcal{T}^{\rho}$ corresponds directly to the fidelity between the resource state $\rho$ and the maximally entangled state~$\Phi_n$:
\begin{align}\label{eq:entanglement_fidelity_teleportation}
 F(\mathcal{T}^{\rho}) = \braket{\Phi_n|\rho|\Phi_n}.
\end{align}
Consequently, improving the fidelity $\braket{\Phi_n|\rho|\Phi_n}$ of the resource state $\rho$ directly enhances the entanglement fidelity of the teleportation channel.

However, attempts to improve this resource fidelity for teleportation are constrained by the requirement that only LOCC are permitted.
Fundamentally, LOCC cannot increase the entanglement of the shared state $\rho_{AB}$~\cite{Chitambar2014}. 
Given that the target state $\Phi_n$ is maximally entangled, if $\rho_{AB}$ initially possesses less entanglement than $\Phi_n$, then the fidelity $\braket{\Phi_n|\Lambda(\rho_{AB})|\Phi_n}$ achievable with any LOCC operation $\Lambda \in\locc(A,B)$ will inherently be less than one.
This motivates the \emph{fidelity of distillation} $f(\rho_{AB})$, which assesses the potential quality of a resource state for tasks like teleportation~\cite{Regula2019,Regula2020}.
It quantifies the maximal achievable fidelity with the target state $\Phi_n$ starting from $\rho_{AB}$, optimized over all possible LOCC channels $\Lambda \in\locc(A,B)$: 
\begin{align}\label{eq:fid_of_distillation}
 f(\rho_{AB}) = \max_{\Lambda \in\locc(A,B)} \braket{\Phi_n|\Lambda(\rho_{AB})|\Phi_n}.
\end{align}
The value of $f(\rho_{AB})$ ranges from $2^{-n}$ for separable states to $1$ for maximally entangled states~\cite{Horodecki2007}. 
States are classified as NME when $2^{-n} < f(\rho_{AB}) < 1 $.

\subsection{Channel twirling}\label{sec:twirl}
Twirling is a mathematical concept that aims to simplify the structure of a quantum channel by averaging it over a set of unitary transformations.
The full channel twirl of a channel $\mathcal{C}$ over the set of $n$-qubit unitary operators $\mathcal{U}(2^n)$ is given by
\begin{align}
 \mathbb{E}_\mu(\mathcal{C}): \rho \mapsto \int_{\mathcal{U}(2^n)} U^\dagger\mathcal{C}(U\rho U^{\dagger})U \,d\mu(U),
\end{align}
where the average is calculated according to the Haar measure $\mu$ on $\mathcal{U}(2^n)$.
For any channel $\mathcal{C}$, the result of full twirling is always a depolarizing channel: $\mathbb{E}_\mu(\mathcal{C})=\mathcal{D}_{F(\mathcal{C})}$~\cite{Horodecki1999}. 
Importantly, the entanglement fidelity is preserved under twirling: $F(\mathcal{C}) = F(\mathbb{E}_\mu(\mathcal{C}))$~\cite{Horodecki1999}.

In practice, the integral over the Haar measure is approximated by averaging over finite ensembles of unitary operators. 
For an ensemble $\mathcal{E} = \{(p_i, U_i)\}_{i=0}^{K-1}$ with $U_i\in \mathcal{U}(2^n)$ and probabilities $p_i$, we denote an $\mathcal{E}$-channel-twirl of $\mathcal{C}$ by
\begin{align}\label{eq:expectation_ensemble}
 \mathbb{E}_{\mathcal{E}}(\mathcal{C}): \rho \mapsto \sum_{i=0}^{K-1} p_iU_i^\dagger\mathcal{C}(U_i\rho U_i^{\dagger})U_i.
\end{align}

An ensemble is called a \emph{unitary two-design}~\cite{Cleve2016} if the $\mathcal{E}$-channel-twirl exactly reproduces the full twirl of $\mathcal{C}$:
\begin{align}
\mathbb{E}_{\mathcal{E}}(\mathcal{C}) = \mathbb{E}_\mu(\mathcal{C}).
\end{align}
For $n$-qubit systems, the minimal ensemble size of a unitary two-design satisfies~\cite{Gross2007, Roy2009}
\begin{align}\label{eq:size_two_design}
 2^{4n} - 2^{2n+1} + 2 \le K \le 2^{5n} - 2^{3n}.
\end{align}
For one qubit, the following ensemble of $12$ unitaries with equal probability is a two-design~\cite{Bengtsson2017}
\begin{align}\label{eq:single_qubit_two_desing}
 \begin{split}
 \big\{(12^{-1}, AB)\big|& A \in \{I, HS, SH\}, \\
 &B\in\{I, X, Y, Z\}\big\}.
 \end{split}
\end{align}
For $n$ qubits, there exist efficient constructions for unitary two-designs~\cite{Cleve2016,Dankert2009}.

Another property of an ensemble is \emph{Pauli mixing}~\cite{Cleve2016}.
Such an ensemble $\mathcal{E} = \{(p_i, U_i)\}_{i=0}^{K-1}$ comprises unitaries $U_i$ that preserve the $n$-qubit Pauli group $\mathcal{P}_n$ under conjugation, i.e., for any $P \in \mathcal{P}_n$, the operator $U_i^{\dagger}PU_i$ is also an element of $\mathcal{P}_n$.
This ensemble is then defined as Pauli mixing if, for any $P \in \mathcal{Q}_n^*$, the random output $P_{\text{out}} =\pi(U^{\dagger}PU)$, where $U$ is drawn from $\mathcal{E}$ with probability $p_i$, is uniformly distributed over $\mathcal{Q}_n^*$. 
Here, $\pi$ is the projection defined in \Cref{eq:projection_pauli}.
This uniform distribution means that for any input $P \in \mathcal{Q}_n^*$ and any target $P' \in \mathcal{Q}_n^*$, the probability of $P_{\text{out}}$ being $P'$ is:
\begin{equation}\label{eq:definition_pauli_mixing}
\Pr(P_{\text{out}} = P') = \sum_{i \in \text{Idx}(P, P')} p_i = \frac{1}{|\mathcal{Q}_n^*|},
\end{equation}
where the index set is $\text{Idx}(P, P') := \{i \mid \pi(U_i^{\dagger}PU_i) = P'\}$.
Such an ensemble thereby effectively randomizes any non-identity Pauli operator into a uniform probabilistic mixture over $\mathcal{Q}_n^*$.
As a result, twirling any any Pauli channel $\mathcal{C}_{\text{Pauli}}$ with a Pauli mixing ensemble $\mathcal{E}$ transforms it into a depolarizing channel that retains the original entanglement fidelity:
\begin{align}\label{eq:pauli_mixing_to_depol}
\mathbb{E}_{\mathcal{E}}(\mathcal{C}_{\text{Pauli}}) = \mathcal{D}_{F(\mathcal{C}_{\text{Pauli}})}.
\end{align}
A proof is provided by \Cref{lemma_pauli_mixing} in \Cref{sec:appendix_pauli_mixing}.
Pauli mixing ensembles can be smaller than unitary two-designs.
For $n>1$ qubits, there exists a Pauli mixing ensemble of size $2^{3n} - 2^{n}$, which can be efficiently constructed~\cite{Cleve2016}.
In the special case of $n=1$, three unitaries suffice for a Pauli mixing ensemble~\cite{Chau2005}:
\begin{align}\label{eq:single_qubit_pauli_mixing}
 \left\{(3^{-1},U)\middle| U \in \{I, HS, SH\}\right\}.
\end{align}

\subsection{Simulation of state transfer via a QPD}
The operations physically realizable with current hardware between devices $A$ and $B$ are typically limited to a set $S \subset \cptp(A, B)$. 
Given this limitation, the goal is to simulate the ideal, noiseless state transfer from $A$ to $B$ using only the available operations from the set $S$.
A common example for $S$ involves operations derived from $\locc(A,B)$.
The perfect state transfer, ensuring faithful transmission of an $n$-qubit quantum state, is represented by the identity channel $\mathcal{I}^{\otimes n}_{A \rightarrow B}$.
The QPD expresses this ideal channel as a linear combination of implementable channels $\mathcal{F}_i \in S$:
\begin{align}\label{eq:qpd_identity}
 \mathcal{I}^{\otimes n}_{A \rightarrow B} = \sum_{i=1}^m c_i \mathcal{F}_i.
\end{align}
The coefficients $\{c_i\}_{i}$ form a quasiprobability distribution, satisfying $\sum_i c_i = 1$ while allowing for negative values $c_i \in \mathbb{R}$.

The presence of negative coefficients in the QPD prevents a direct physical realization of $\mathcal{I}^{\otimes n}_{A \rightarrow B} $ as a probabilistic mixture of the channels $\mathcal{F}_i \in S$.
However, the QPD still allows the exact calculation of expectation values for the ideal channel.
Using \Cref{eq:qpd_identity}, the expectation value $\tr[O\mathcal{I}^{\otimes n}_{A \rightarrow B}(\rho)]$, where the state $\rho$ is transferred from system $A$ to $B$ and then measured with observable $O$, can be formulated as:
\begin{align}
 \tr[O\mathcal{I}^{\otimes n}_{A \rightarrow B}(\rho)] &= \sum_{i=1}^m c_i \tr[O\mathcal{F}_i(\rho)]\\
 &= \sum_{i=1}^m p_i \tr[O\mathcal{F}_i(\rho)] \kappa \sign(c_i) \label{eq:expectation}
\end{align}
where $p_i= |c_i|\kappa^{-1}$ and $\kappa = \sum_i |c_i| \ge 1$.

This formulation enables the quasiprobabilistic simulation of the state transfer.
This is a Monte Carlo simulation to estimate the target expectation value $\tr[O\mathcal{I}^{\otimes n}_{A \rightarrow B}(\rho)]$ by probabilistically sampling only operations $\mathcal{F}_{i}$ from the set $S$.
The protocol involves $N$ independent circuit executions, where each execution $j$ proceeds as follows: 
First, an index $1\le i^{(j)} \le m$ is selected according to the probability distribution $p_i$. 
The corresponding operator $\mathcal{F}_{i^{(j)}}$ is then applied to the input state $\rho\in D(A)$, preparing the modified state $\mathcal{F}_{i^{(j)}}(\rho) \in D(B)$. 
Next, the observable $O$ is measured on this state, yielding a single-shot measurement outcome~$m_j\in \mathbb{R}$.
Finally, each outcome $m_j$ is classically scaled by $\kappa \sign(c_{i(j)})$, which is the corresponding coefficient to the chosen $\mathcal{F}_{i^{(j)}}$.
The final expectation value estimate is computed by averaging all weighted outcomes:
\begin{align}\label{eq:qpd_estimator}
\widehat{\braket{O}}_\rho^N = \frac{1}{N}\sum_{j=0}^{N-1}\kappa \sign(c_{i^{(j)}}) m_{j}.
\end{align}
As $N \to \infty$, this estimator converges to the true expectation value $\tr[O\mathcal{I}^{\otimes n}_{A \rightarrow B}(\rho)]$.
However, to achieve a target statistical error $\epsilon$ with this quasiprobabilistic simulation, the number of required shots $N$ scales as $O(\kappa^2/\epsilon^2)$~\cite{Temme2017}. 
This introduces a sampling overhead of $O(\kappa^2)$ compared to directly estimating the expectation value $\tr[O\mathcal{I}^{\otimes n}_{A \rightarrow B}(\rho)]$ from a hypothetical, direct implementation of the ideal state transfer.
Minimizing the QPD’s $\kappa$-factor is therefore critical for practical implementations.

The QPD framework for simulating ideal identity channels is applied in circuit cutting techniques such as wire cutting~\cite{Brenner2023,Harada2024,Lowe2022}.
In wire cutting, a direct $n$-qubit quantum connection, i.e., an identity channel $\mathcal{I}^{\otimes n}_{A \rightarrow B}$ that ensures faithful state transfer between circuit parts $A$ and $B$, is replaced by its quasiprobabilistic simulation. 
The objective is to replicate this removed identity channel's action, typically using only LOCC between $A$ and $B$. 
Such LOCC protocols usually entail measurements on system $A$, classical communication of the outcomes, and conditioned state preparations on system $B$.
The minimal possible sampling overhead $\gamma(\mathcal{I}^{\otimes n}_{A \rightarrow B})$ of such a wire cut, characterized by the minimal $\kappa$-factor of all possible QPDs, is known to be~\cite{Brenner2023}:
\begin{align}\label{eq:overhead_trad_wire_cut}
\gamma(\mathcal{I}^{\otimes n}_{A \rightarrow B}) = 2^{n+1} - 1.
\end{align}

If the devices additionally share a NME state $\rho_{A'B}$ as a resource, where $A'$ is an auxiliary system held by party $A$, the range of implementable operations effectively expands. 
The LOCC protocols can now leverage this shared entanglement. 
In this entanglement-assisted scenario, the minimal sampling overhead $\gamma^{\rho}(\mathcal{I}^{\otimes n}_{A \rightarrow B})$ of the wire cut is reduced and depends on the shared resource state via its fidelity of distillation~$f(\rho_{A'B})$~\cite{Bechtold2024a,Bechtold2025}:
\begin{align}\label{eq:gamma_nme_simplified}
\gamma^{\rho}(\mathcal{I}^{\otimes n}_{A \rightarrow B}) = \frac{2}{f(\rho_{A'B})} - 1.
\end{align}

For separable states $\rho_{A'B}$, the fidelity of distillation evaluates to $f(\rho_{A'B})= 2^{-n}$, which recovers the overhead without NME states $\gamma^{\rho}(\mathcal{I}^{\otimes n}_{A \rightarrow B})=\gamma(\mathcal{I}^{\otimes n}_{A \rightarrow B})$.
Increasing the entanglement of the shared NME state increases $f(\rho_{A'B})$, thereby reducing the sampling overhead $\gamma^{\rho}(\mathcal{I}^{\otimes n}_{A \rightarrow B})$.
As $\rho_{A'B}$ approaches a maximally entangled state, i.e., $f(\rho_{A'B}) \rightarrow 1$, the minimal sampling overhead $\gamma^{\rho}(\mathcal{I}^{\otimes n}_{A \rightarrow B})$ converges to one, eliminating the need for additional circuit executions.
This ideal scenario is equivalent to using the maximal entangled resource state in quantum teleportation.

 \section{Simulating ideal state transfers using a noisy state transfer}\label{sec:main}
This section presents our method for simulating ideal quantum state transfers across distributed quantum devices using imperfect physical channels.
First, \Cref{sec:main_qpd} introduces a QPD of the ideal identity channel $\mathcal{I}^{\otimes n}$.
This decomposition leverages the available imperfect channel between devices to simulate ideal state transfers, exhibiting a sampling overhead that decreases as the channel quality improves. 
To enable practical implementation, \Cref{sec:comp_ent_fid} presents a systematic method for calculating the required QPD coefficients for any given channel between the devices.
Furthermore, \Cref{sec:wire_cut_teleportation} applies the developed QPD to formulate a wire cutting protocol that uses shared NME states.
Significantly, this approach can leverage arbitrary NME states, including mixed ones, overcoming the restriction to idealized pure states found in prior protocols~\cite{Bechtold2024a,Bechtold2025}.
Finally, \Cref{sec:error_analysis} analyzes potential errors that may arise during the practical implementation of this QPD.

\subsection{QPD for the identity channel}\label{sec:main_qpd}
The goal is to enable the simulation of the $n$-qubit identity channel $\mathcal{I}^{\otimes n}$ using operations derived from an available, generally imperfect, physical channel $\mathcal{C}$ shared between distributed devices, supplemented by LOCC.
Our approach relies on expressing $\mathcal{I}^{\otimes n}$ as a linear combination of two specific $n$-qubit depolarizing channels: $\mathcal{D}_p$ with entanglement fidelity given by parameter $p\ne 0$ and $\mathcal{D}_0$ with zero entanglement fidelity.
By rearranging the definition of the depolarizing channel as given in \Cref{eq:depol_channel_with_D0}, the identity relating these channels is:
\begin{align}\label{eq:qpd_identiy_depol}
 \mathcal{I}^{\otimes n} = \frac{1}{p}\mathcal{D}_p - \left(\frac{1}{p}- 1\right)\mathcal{D}_0.
\end{align}
This decomposition forms the basis of our QPD. 
It allows simulating $\mathcal{I}^{\otimes n}$ by probabilistically executing either $\mathcal{D}_p$ or $\mathcal{D}_0$, weighted by the potentially negative coefficients $1/p$ and $-(1/p-1)$, respectively. 
The feasibility of this simulation hinges on our ability to construct $\mathcal{D}_p$ and $\mathcal{D}_0$ using the available physical resources, i.e., the shared channel $\mathcal{C}$ and LOCC.
We now detail the construction of these two depolarizing channels.

The channel $\mathcal{D}_p$ required in \Cref{eq:qpd_identiy_depol} can be obtained by twirling the available shared channel $\mathcal{C}$ with a suitable ensemble of unitary operations $\mathcal{E}=\{(p_i, U_i)\}$ such that the $\mathcal{E}$-channel-twirl yields $\mathbb{E}_{\mathcal{E}}(\mathcal{C}) = \mathcal{D}_{F(\mathcal{C})}$.
For instance, while a unitary two-design is generally suitable for an arbitrary channel $\mathcal{C}$, simpler ensembles can be employed if $\mathcal{C}$ possesses known structural properties.
These adaptations, aimed at minimizing the number of operators in the ensemble, will be further elaborated towards the end of this subsection.
Regardless of the specific ensemble $\mathcal{E}$ that achieves the desired twirling outcome, this process results in a depolarizing channel whose parameter $p$ is precisely the entanglement fidelity $F(\mathcal{C})$ of the original channel $\mathcal{C}$. 
Thus, we identify $p = F(\mathcal{C})$ in \Cref{eq:qpd_identiy_depol}.

The depolarizing channel $\mathcal{D}_0$ with zero entanglement fidelity must also be synthesized from available operations. 
While it can be constructed by twirling an arbitrary channel with zero entanglement fidelity, we aim to use a resource-efficient approach that requires the ensemble for the twirl as small as possible.
Specifically, we construct $\mathcal{D}_0$ by first implementing a measure-and-prepare channel $\mathcal{M}$ requiring only one-way classical communication, and then twirling $\mathcal{M}$ with a specific unitary ensemble $\mathcal{V}$. 
The following lemma details this construction.
\begin{restatable}{lemma}{restatelemma}
\label{lemma_mp_channel}
 Let $V_j$ be the $2^n+1$ different unitary operators that diagonalize the sets of maximally commuting Pauli operators, as defined in \Cref{eq:S_j}.
 Let $\mathcal{V} = \{(2^n+1)^{-1}, V_j^\dagger\}_{j=0}^{2^n}$ be the uniform ensemble over the adjoint unitary operators. 
 The $n$-qubit depolarizing channel $\mathcal{D}_0$ can be constructed as:
 \begin{align}
 \mathcal{D}_0 = \mathbb{E}_{\mathcal{V}}(\mathcal{M}),
 \end{align}
 where $\mathcal{M}$ is a measure-and-prepare channel defined by measuring in the computational basis $\{\ket{\vec{k}}\}_{\vec{k}\in\{0,1\}^n}$ and preparing a state $\rho_{\vec{k}}$ (see \Cref{eq:state_rho_k}) upon outcome~$\ket{\vec{k}}$:
 \begin{align}
 \mathcal{M}: \rho \mapsto \sum_{\vec{k}\in\{0,1\}^n}\tr\left[\ketbras{\vec{k}}\rho\right]\rho_{\vec{k}}
 \end{align}
 with the prepared states $\rho_{\vec{k}}$ being uniform mixtures over all computational basis states orthogonal to $\ket{\vec{k}}$, where $\delta_{\vec{k},\vec{l}}$ is the Kronecker delta:
 \begin{align}\label{eq:state_rho_k}
 \rho_{{\vec{k}}} = \sum_{\vec{l}\in\{0,1\}^n} \frac{1-\delta_{\vec{k},\vec{l}}}{2^n-1}\ketbras{\vec{l}}.
 \end{align}
 This construction uses the minimal number of measure-and-prepare circuits required to realize $\mathcal{D}_0$ without ancilla qubits.
\end{restatable}
A formal proof of the lemma is provided in \Cref{sec:proof_lemma}.
Substituting the constructions $\mathcal{D}_p = \mathbb{E}_{\mathcal{E}}(\mathcal{C})$ with $p = F(\mathcal{C})$ and $\mathcal{D}_0 = \mathbb{E}_{\mathcal{V}}(\mathcal{M})$ into the initial decomposition of \Cref{eq:qpd_identiy_depol} yields the central theorem for the QPD of the identity channel.

\begin{figure*}
 \centering
\includegraphics[trim={0 -0.4cm 0 -0.4cm},clip]{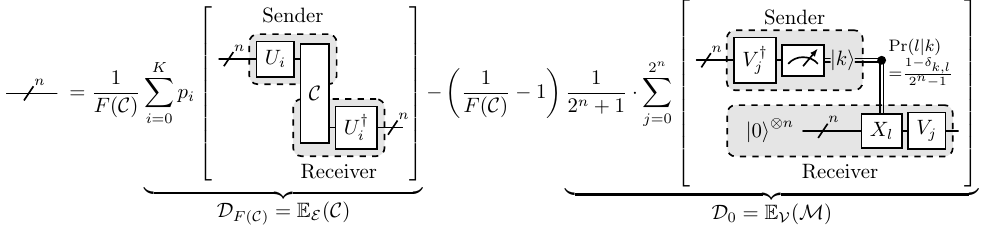}
 \caption{Circuit representation of the QPD from \Cref{theorem_arbitrary_cut}. The first term involves applying the channel $\mathcal{C}$ between sender and receiver, while the second term utilizes measure-and-prepare circuits, where $\Pr(l|k)$ denotes the conditional probability of applying operation $X_l$ (as defined in \Cref{eq:X_and_Z}) given measurement outcome $\ket{k}$.}
 \label{fig:wire_cut_n_qubits}
\end{figure*}

\begin{theorem}\label{theorem_arbitrary_cut}
 Let $\mathcal{C}$ be an $n$-qubit quantum channel with entanglement fidelity $F(\mathcal{C})$. 
 Let $\mathcal{E} = \{p_i, U_i\}_{i=0}^{K-1}$ be a unitary ensemble such that its $\mathcal{E}$-channel-twirl yields the depolarizing channel $\mathbb{E}_{\mathcal{E}}(\mathcal{C})=\mathcal{D}_{F(\mathcal{C})}$. 
 Let $\mathcal{V}$ and $\mathcal{M}$ be the unitary ensemble and measure-and-prepare channel defined in \Cref{lemma_mp_channel}. 
 Then, the $n$-qubit identity channel $\mathcal{I}^{\otimes n}$ admits the following QPD:
 \begin{equation}
 \mathcal{I}^{\otimes n} = \frac{1}{F(\mathcal{C})}\mathbb{E}_{\mathcal{E}}(\mathcal{C}) - \left(\frac{1}{F(\mathcal{C})}-1\right)\mathbb{E}_{\mathcal{V}}(\mathcal{M}).
 \end{equation}
\end{theorem}
\begin{proof}
 The theorem follows directly by substituting $p=F(\mathcal{C})$, $\mathcal{D}_p = \mathbb{E}_{\mathcal{E}}(\mathcal{C})$, and $\mathcal{D}_0 = \mathbb{E}_{\mathcal{V}}(\mathcal{M})$ into \Cref{eq:qpd_identiy_depol}, using the result from \Cref{lemma_mp_channel}.
\end{proof}
\Cref{fig:wire_cut_n_qubits} shows a circuit implementation of the QPD of \Cref{theorem_arbitrary_cut}.
Executing the QPD involves the twirled channels $\mathbb{E}_{\mathcal{E}}(\mathcal{C})$ and $ \mathbb{E}_{\mathcal{V}}(\mathcal{M})$. 
Their implementation first requires sampling a unitary operator $U_i$ from $\mathcal{E}$ or $V_j$ from $\mathcal{V}$. 
Based on the sampled unitary operator, for $\mathbb{E}_{\mathcal{E}}(\mathcal{C})$, the sender then locally applies $U_i$ before the channel $\mathcal{C}$, and the receiver applies $U_i^\dagger$ after it. 
Similarly, for $\mathbb{E}_{\mathcal{V}}(\mathcal{M})$, the sender applies $V_j^{\dagger}$ before channel $\mathcal{M}$, and the receiver applies $V_j$ after it. 
Since these local operations are performed by spatially separated parties (sender and receiver) but must correspond to the same sampled operator ($U_i$ or $V_j$), their actions require coordination. 
This coordination is typically achieved by classically communicating the selected operator from the sender or by using pre-shared randomness to determine the sequence of unitary operators in advance.

The sampling overhead for the QPD presented in \Cref{theorem_arbitrary_cut} is determined by the sum of the absolute values of its quasiprobability coefficients. 
For this specific decomposition, the overhead $\kappa$ is given by:
\begin{align}\label{eq:overhead_qpd_channel}
 \kappa = \left|\frac{1}{F(\mathcal{C})}\right|+\left|- \left(\frac{1}{F(\mathcal{C})}-1\right)\right|=\frac{2}{F(\mathcal{C})} -1.
\end{align}
Thus, by using a channel $\mathcal{C}$ with higher entanglement fidelity $F(\mathcal{C})$, the sampling overhead of the QPD decreases, reaching its minimum when $\mathcal{C}$ approaches the ideal identity channel $\mathcal{I}^{\otimes n}$ with $F(\mathcal{I}^{\otimes n}) = 1$.
In this case, the identity channel $\mathcal{I}^{\otimes n}$ can be used directly without sampling overhead.
Crucially, for the sampling overhead of the QPD from \Cref{theorem_arbitrary_cut} as given in \Cref{eq:overhead_qpd_channel} to be lower than that of an optimal wire cut with only classical communication (see \Cref{eq:overhead_trad_wire_cut}), the entanglement fidelity of the used channel $\mathcal{C}$ must satisfy:
\begin{align}\label{eq:advantage_qpd_over_wire_cut}
    F(\mathcal{C}) > 2^{-n}.
\end{align}

Beyond reducing the sampling overhead, it is also practical to minimize the number of distinct circuits required when sampling from the QPD of \Cref{theorem_arbitrary_cut}. 
This can be achieved by carefully selecting the unitary ensemble $\mathcal{E}$ used for twirling $\mathcal{C}$.
While the ensemble $\mathcal{V}$ for twirling $\mathcal{M}$ is already minimal and fixed by \Cref{lemma_mp_channel}, the ensemble $\mathcal{E}$ should contain the minimum number of unitaries necessary to achieve the required twirling result, i.e., $\mathbb{E}_{\mathcal{E}}(\mathcal{C})=\mathcal{D}_{F(\mathcal{C})}$. 
The choice of the minimal ensemble $\mathcal{E}$ depends on the known structural properties of the channel $\mathcal{C}$:
\begin{enumerate}
\item For arbitrary channels $\mathcal{C}$, the ensemble $\mathcal{E}$ must generally be a unitary two-design to guarantee twirling to a depolarizing channel.
\item If $\mathcal{C}$ is known to be a Pauli channel, a smaller Pauli-mixing ensemble suffices for $\mathcal{E}$, as this directly transforms Pauli channels into depolarizing ones (see \Cref{eq:pauli_mixing_to_depol}). 
\item If $\mathcal{C}$ is already a depolarizing channel, no twirling is needed. 
The trivial ensemble $\mathcal{E}=\{(1, I^{\otimes n})\}$ can be used.
\end{enumerate}
This set of specific conditions allows optimizing the QPD's circuit count based on prior knowledge about $\mathcal{C}$.
Furthermore, for specific channels $\mathcal{C}$ with particular symmetries, even smaller unitary ensembles tailored to their structure might exist that satisfy the twirling condition, further reducing the number of circuits.

\subsection{Calibrating the QPD's coefficients}\label{sec:comp_ent_fid}
The QPD presented in \Cref{theorem_arbitrary_cut} utilizes an arbitrary shared channel $\mathcal{C}$ as a resource to reduce sampling overhead.
However, practical application of this QPD requires its calibration to the specific channel $\mathcal{C}$ by adapting the QPD's coefficients. 
These coefficients are crucial as they dictate the sampling probabilities for the various channels within the QPD.
The values of these coefficients, in turn, depend directly on the channel's entanglement fidelity $F(\mathcal{C})$.
Therefore, determining $F(\mathcal{C})$ is a prerequisite for implementing the quasiprobabilistic simulation. 
This subsection outlines a practical method for measuring $F(\mathcal{C})$ on distributed quantum devices connected by channel $\mathcal{C}$ and capable of LOCC.

Directly measuring the entanglement fidelity, as defined in \Cref{eq:entanglement_fidelity}, poses practical challenges in the considered distributed scenario.
It would require preparing a maximally entangled state $\ket{\Phi_n}$ on one device, transmitting half of it via the channel $\mathcal{C}$ to the other device, and subsequently performing a joint measurement projecting the resulting shared state onto $\ket{\Phi_n}$. 
Such non-local measurements are unavailable in the considered distributed scenario.

Our approach circumvents this difficulty by leveraging the $\mathcal{E}$-channel-twirl that is part of the QPD of \Cref{theorem_arbitrary_cut}.
Twirling the channel $\mathcal{C}$ with this unitary ensemble $\mathcal{E}$ yields the depolarizing channel $\mathbb{E}_{\mathcal{E}}(\mathcal{C})=\mathcal{D}_{F(\mathcal{C})}$.
The entanglement fidelity $F(\mathcal{C})$ can be directly related to the probability $P_{0 \to 0}$ of preserving the specific state $\ket{0}^{\otimes n}$ when passed through this depolarizing channel.
This relationship, whose derivation is detailed in \Cref{sec:appendix_ent_fidelity}, is given by:
\begin{align}\label{eq:compute_fid_ent_ordered}
F(\mathcal{C}) = \frac{2^n+1}{2^n} P_{0 \to 0} - \frac{1}{2^n},
\end{align}
where $P_{0 \to 0}$ is defined as
\begin{align}\label{eq:P_0_to_0}
P_{0 \to 0} = \braket{0^{\otimes n}|\mathbb{E}_{\mathcal{E}}(\mathcal{C})((\ketbras{0})^{\otimes n})|0^{\otimes n}}.
\end{align}

To calculate $F(\mathcal{C})$ using this expression, the probability $P_{0 \to 0}$ must be experimentally determined.
This measurement uses the channel $\mathcal{C}$ and either classical communication or pre-shared randomness to coordinate the channel twirl.
Crucially, this characterization of $F(\mathcal{C})$ typically needs to be performed only once initially, before executing subsequent quasiprobabilistic simulations involving the QPD with $\mathcal{C}$. 
This relies on the assumption that the noise characteristics of the shared channel are stable over the relevant timescale.

\subsection{Wire cutting with NME states}\label{sec:wire_cut_teleportation}
While the QPD detailed in \Cref{theorem_arbitrary_cut} employs an arbitrary shared channel $\mathcal{C}$ as a resource for simulating the identity channel, shared entangled states are often the primary resource for distributed quantum computing~\cite{Barral2024,Caleffi2024}.
This section adapts our QPD framework to utilize such shared entangled states, specifically connecting this approach to wire cutting techniques that employ potentially mixed NME states.
A key advantage of this adaptation is its ability to handle mixed NME states, thereby addressing limitations of previous work restricted to pure NME resources~\cite{Bechtold2024a,Bechtold2025}.
Moreover, shared entanglement offers greater flexibility, as it can be established between distant devices indirectly, e.g., via entanglement routing across intermediate nodes~\cite{Abane2025}.
Furthermore, as \Cref{sec:channel_to_pauli_channel} explains, the ability to convert any available channel $\mathcal{C}$ into a shared entangled state reinforces the focus on entanglement, particularly since this approach may yield greater benefits than directly applying the channel within the QPD.

The main idea is to implement the required channel $\mathcal{C}$ from the QPD of \Cref{theorem_arbitrary_cut} by utilizing the available shared NME state $\rho$ as a resource within an LOCC protocol. 
A primary example, detailed below, is realizing $\mathcal{C}$ via the standard quantum teleportation protocol $\mathcal{T}^{\rho}$, which consumes the state $\rho$.

\subsubsection{Sampling overhead}
Implementing the QPD using $\mathcal{C}=\mathcal{T}^{\rho}$ incurs a sampling overhead determined by the entanglement fidelity of the teleportation channel $F(\mathcal{T}^{\rho})$.
According to \Cref{eq:entanglement_fidelity_teleportation}, the entanglement fidelity is given by $F(\mathcal{T}^{\rho}) = \braket{\Phi_{n}|\rho|\Phi_{n}}$.

Therefore, to minimize the sampling overhead associated with a given initial NME state $\rho$, one must maximize its fidelity with $\Phi_n$ before the state is used in teleportation.
Using the idea of the fidelity of distillation, as defined in \Cref{eq:fid_of_distillation}, this maximization is achieved by preprocessing the state $\rho$ with the optimal LOCC operator $\Lambda_{\max}$ that yields the highest possible fidelity:
\begin{align}
 \Lambda_{\max} = \argmax_{\Lambda \in\locc(A,B)} \braket{\Phi_n|\Lambda(\rho)|\Phi_n}.
\end{align}

Let $f(\rho)=F(\mathcal{T}^{\Lambda_{\max}(\rho)}) = \braket{\Phi_n|\Lambda_{\max}(\rho)|\Phi_n}$ denote this maximum LOCC-assisted fidelity achievable from the initial state $\rho$. 
By implementing the channel as $\mathcal{C} = \mathcal{T}^{\Lambda_{\max}(\rho)}$ in \Cref{theorem_arbitrary_cut}, the wire cutting protocol achieves the minimal possible sampling overhead according to \Cref{eq:gamma_nme_simplified} obtainable from the initial NME state $\rho$:
\begin{align}
 \frac{2}{F(\mathcal{T}^{\Lambda_{\max}(\rho)})} -1 &= \frac{2}{f(\rho)} -1
 = \gamma^{\rho}(\mathcal{I}^{\otimes n}_{A \rightarrow B}).
\end{align}
Consequently, we obtain a wire cutting protocol capable of achieving the optimal sampling overhead when using a shared NME state $\rho$, subject to optimizations with LOCC operators $\Lambda$ on the state $\rho$.

However, finding the optimal $\Lambda_{\max}$ presents a significant challenge. 
The complete set of LOCC protocols allows for an unbounded number of communication rounds~\cite{Chitambar2014}, making a direct maximization over all possible operations from this set intractable on physical hardware.
Instead, a promising strategy to identify feasible LOCC operations with limited communication yielding high fidelity could be the application of iterative optimization routines.
Such routines would optimize over restricted, parameterizable subsets of LOCC operations that are feasible for the target devices. 
The optimization aims to maximize the entanglement fidelity $F(\mathcal{T}^{\Lambda(\rho)})$, which can be estimated using the technique described in \Cref{sec:comp_ent_fid}. 
For instance, a practical compromise is optimizing over local unitary operators without classical communication or only a single communication round, balancing performance with practical feasibility on near-term hardware.

\subsubsection{Transforming a channel into a Pauli channel}\label{sec:channel_to_pauli_channel}
Even when provided with an arbitrary quantum channel $\mathcal{C}$ between distributed devices, it can be advantageous to first convert it into an entangled state for use in wire cutting. 
The following conversion procedure simplifies the twirling channel within the QPD of \Cref{theorem_arbitrary_cut}: the required ensemble $\mathcal{E}$ reduces from a general unitary two-design to a smaller Pauli-mixing ensemble, critically, without increasing the sampling overhead.

The procedure involves using the shared channel $\mathcal{C}$ to generate its corresponding Choi state $J(\mathcal{C})$~\cite{Stratton2024}, shared between the devices:
\begin{align}\label{eq:choi_state}
 J(\mathcal{C}) = (\mathcal{I}^{\otimes n} \otimes \mathcal{C} )(\Phi_n).
\end{align}
When this Choi state $J(\mathcal{C})$ is subsequently used as the resource for the standard teleportation protocol, the resulting channel $\mathcal{T}^{J(\mathcal{C})}$ is given by:
\begin{align}
 \mathcal{T}^{J(\mathcal{C})}(\varphi) = \sum_{a=0}^{2^{2n} - 1}\braket{\Phi^{a}|J(\mathcal{C})|\Phi^{a}}P_a\varphi P_a.
\end{align}
Assuming no error in the teleportation process itself, the resulting channel is a Pauli channel and a Pauli twirling ensemble suffices for the QPD.

Crucially, this transformation preserves the entanglement fidelity of the original channel $\mathcal{C}$:
\begin{align}
 F(\mathcal{T}^{J(\mathcal{C})}) &= \braket{\Phi_{n}|J(\mathcal{C})|\Phi_{n}} \label{eq:fidelity_preservation_steps_1}\\
 &=F(\mathcal{C}). \label{eq:fidelity_preservation_steps_2}
\end{align}
The first equality, \Cref{eq:fidelity_preservation_steps_1}, applies \Cref{eq:entanglement_fidelity_teleportation} to connect $F(\mathcal{T}^{J(\mathcal{C})})$ with the fidelity of the resource state $J(\mathcal{C})$. 
The second equality, \Cref{eq:fidelity_preservation_steps_2}, subsequently follows from the definitions of the Choi state $J(\mathcal{C})$ (via \Cref{eq:choi_state}) and the entanglement fidelity $F(\mathcal{C})$ (as per \Cref{eq:entanglement_fidelity}).
Consequently, because the sampling overhead is determined by this entanglement fidelity, it remains unchanged under this transformation.
Therefore, by setting $\mathcal{C}' = \mathcal{T}^{J(\mathcal{C})}$ in the QPD of \Cref{theorem_arbitrary_cut}, one can leverage the practical benefit of using a smaller Pauli-mixing ensemble $\mathcal{E}$ while maintaining the same sampling overhead determined by the original channel's fidelity $F(\mathcal{C})$.

\subsection{Error analysis for \Cref{theorem_arbitrary_cut}}\label{sec:error_analysis}
In the following, we analyze error sources when sampling from the QPD of \Cref{theorem_arbitrary_cut} and building an estimator as shown in \Cref{eq:qpd_estimator}.
In practice, implementing the QPD is subject to noise and errors.
Therefore, it does not yield the perfect identity channel $\mathcal{I}$, but rather an imperfect channel realization denoted by $\mathcal{\tilde{I}}$:
\begin{align}
 \mathcal{\tilde{I}} = \frac{1}{\tilde{p}}\tilde{\mathcal{D}}_p - \left(\frac{1}{\tilde{p}}-1\right)\tilde{\mathcal{D}}_0
\end{align}
where $\tilde{\mathcal{D}}_{p}$ and $\tilde{\mathcal{D}}_{0}$ represent potentially imperfect implementations of the target depolarizing channels, and $\tilde{p}$ is an non-zero estimate of the ideal depolarization parameter $p$.
A detailed derivation of this expression and subsequent bounds can be found in \Cref{sec:appendix_error}.

The total error $\epsilon(N)$ measures the deviation between the expectation value obtained from $N$ samples using the biased estimate $\widehat{\braket{O}}_{\tilde{\mathcal{I}}(\rho)}^{N}$ and the ideal value $\tr[O\rho]$:
\begin{align}\label{eq:error}
 \epsilon(N) = \left|\widehat{\braket{O}}_{\tilde{\mathcal{I}}(\rho)}^{N} - \tr[O\rho] \right|.
\end{align}
Two primary sources contribute to this total error:
(i)~statistical sampling error $\epsilon_{\text{sampling}}$ arising from a finite number $N$ of samples from the QPD and (ii)~a systematic bias $\epsilon_{\text{bias}}$ resulting from the discrepancy between the implemented channel $\mathcal{\tilde{I}}$ and the ideal identity channel $\mathcal{I}$. 
The total error is bounded by the sum of these contributions, as the statistical and systematic errors are not necessarily additive and can partially cancel:
\begin{align}
 \epsilon(N) \le \epsilon_{\text{sampling}}(N) + \epsilon_{\text{bias}}.
\end{align}

The sampling error $\epsilon_{\text{sampling}}$ quantifies statistical fluctuation of the finite-sample estimator around the true expectation value produced by the implemented channel $\mathcal{\tilde{I}}$:
\begin{align}
 \epsilon_{\text{sampling}}(N) &= \left|\widehat{\braket{O}}_{\tilde{\mathcal{I}}(\rho)}^{N} - \tr[O\tilde{\mathcal{I}}(\rho)] \right|.
\end{align}
This error decreases with the number of samples $N$ as $\mathcal{O}(\kappa/\sqrt{N})$~\cite{Temme2017}, where $\kappa = 2\tilde{p}^{-1} - 1$ represents the sampling overhead associated with the implemented QPD using the estimated parameter $\tilde{p}$.

In contrast, the systematic bias $\epsilon_{\text{bias}}$ represents a fundamental deviation that persists even with an infinite number of samples.
It quantifies the differences between the expectation value obtained with the implemented channel $\mathcal{\tilde{I}}$ and the ideal expectation value obtained with $\mathcal{I}$:
\begin{align}
\epsilon_{\text{bias}} &= \left|\tr[O\tilde{\mathcal{I}}(\rho)] - \tr[O\mathcal{I}(\rho)]\right|\\
\begin{split}
 &\le \left\|O \right\|_{\infty} \biggl(\frac{1}{\tilde{p}}\left\|\tilde{\mathcal{D}}_p -\mathcal{D}_p \right\|_{\diamond} \\
 &\phantom{\le \left\|O \right\|_{\infty} \biggl(} \, + \left(\frac{1}{\tilde{p}}-1\right)\left\|\mathcal{D}_0 - \tilde{\mathcal{D}}_0 \right\|_{\diamond} \\
 &\phantom{\le \left\|O \right\|_{\infty} \biggl(} \, + 2 \left|\frac{p}{\tilde{p}} - 1\right| \biggr).
\end{split}
\end{align}
where the bound is based on the Schatten $\infty$-norm of the specific observable $O$ and reflects three error components:
Imperfections in the realization of channels $\mathcal{D}_p$ and $\mathcal{D}_0$, measured by $\left\|\tilde{\mathcal{D}}_p -\mathcal{D}_p \right\|_{\diamond}$ and $\left\|\tilde{\mathcal{D}}_0 -\mathcal{D}_0 \right\|_{\diamond}$, respectively, and the parameter mismatch between $\tilde{p}$ and ideal $p$.

Errors in the channel implementations typically stem from two main sources: (i)~inherent noise in the quantum hardware used, and (ii)~deliberate approximations made during the synthesis of the required twirling operations. 
A key example of the latter occurs when approximate unitary two-designs $\mathcal{E}_{\text{approx}}$ are employed to realize the twirling needed to generate the depolarizing channels~\cite{Dankert2009,Harrow2009,Haferkamp2022}.
In such cases, the resulting twirled channel $\mathbb{E}_{\mathcal{E}_{\text{approx}}}(\mathcal{C})$ only approximates the ideal depolarizing channel $\mathcal{D}_{F(\mathcal{C})}$, but crucially, the diamond norm distance between the actual and ideal channel is guaranteed to be bounded by a small value.

Similar approximation errors arise if a smaller Pauli mixing ensemble is applied when the underlying channel $\mathcal{C}$ is not inherently a Pauli channel. 
For an arbitrary quantum channel $\mathcal{C}$ characterized by a $\chi$-matrix, the deviation introduced by Pauli twirling is bounded by the sum of the magnitudes of the off-diagonal elements of $\chi$:
\begin{align}\label{eq:pauli_mixing_coherent_errors}
 \left\|\mathbb{E}_{\mathcal{E}}(\mathcal{C}) -\mathcal{D}_{F(\mathcal{C})} \right\|_{\diamond} \le \sum_{i\ne j}|\chi_{ij}|.
\end{align}
This bound links the approximation error directly to the coherent components of the channel's noise quantified by non-zero $\chi_{ij}$ for $i\neq j$.
This result implies that Pauli mixing is a good approximation if the channel $\mathcal{C}$ being twirled has only small coherent error components. 
Consequently, for the wire cutting protocols using NME states introduced in \Cref{sec:wire_cut_teleportation}, if the implemented teleportation process introduces primarily Pauli noise with minimal coherent errors, a Pauli mixing ensemble suffices for achieving a low systematic bias $\epsilon_{\text{bias}}$ provided $\hat{p}$ and $\hat{\mathcal{D}}_0$ is correct.
This approach circumvents the need to implement a larger unitary two-designs.

 \section{Experiments}\label{sec:experiments}
This section presents experiments evaluating the QPD for state transfer proposed in \Cref{theorem_arbitrary_cut}.
Our primary objectives are:
\begin{itemize}
 \item \textbf{Obj. 1:} Experimentally demonstrate the practical feasibility of implementing the QPD (from \Cref{theorem_arbitrary_cut}) and its associated calibration procedure (from \Cref{sec:comp_ent_fid}) on contemporary quantum hardware.
 \item \textbf{Obj. 2:} Investigate the reduction in sampling overhead, related to the sampling error $\epsilon_{\text{sampling}}$, associated with using quantum channels of higher entanglement fidelity.
 \item \textbf{Obj. 3:} Investigate the influence of different unitary ensembles $\mathcal{E}$ on the systematic error $\epsilon_{\text{bias}}$, to assess the viability of using smaller ensembles for reducing QPD circuits on hardware.
\end{itemize}
We use numerical simulations to benchmark performance and explore objectives 2 and 3 under controlled noise models. 
Subsequently, experiments on superconducting quantum devices address all three objectives, assessing practical feasibility (objective 1) and evaluating performance under realistic hardware noise.
The remainder of this section is structured as follows: we first describe the general experimental setup applicable to both simulations and hardware implementations. 
Subsequently, we detail the numerical simulations and present their results. 
This is followed by the methodology specific to the experiments on quantum devices and the corresponding outcomes. 
All code and data generated for this study are publicly available~\cite{darus}.

\subsection{General experimental setup}
Across all experiments, whether numerical simulations or implementations on quantum hardware, we evaluate the accuracy of our QPD for single-qubit state transfers.
This accuracy is quantified by how closely expectation values of Pauli observables $O \in \{X,Z\}$ are estimated for an initial Haar-random input state $\ket{\psi}$ after it undergoes the state transfer using the QPD.
The QPD estimate $\widehat{\braket{O}}_{\tilde{\mathcal{I}}(\psi)}^{N}$ is obtained by using the estimator defined in \Cref{eq:qpd_estimator} from $N$ total shots. This estimate is compared against the exact, classically computed expectation value $\tr[O\ketbras{\psi}]$, assuming an ideal transfer, to determine the estimation error $\epsilon(N)$ as defined in \Cref{eq:error}.

The QPD implementation from \Cref{theorem_arbitrary_cut} requires specifying two channel twirls, which are associated with unitary ensembles $\mathcal{E}$ and $\mathcal{V}$.
We use the specific single-qubit ensembles introduced in \Cref{sec:twirl} to realize these twirls.
For the ensemble $\mathcal{E}$, applied to the quantum channel $\mathcal{C}$, we test three distinct choices to analyze their impact on performance, directly addressing objective 3:
\begin{enumerate}
 \item Unitary two-design from \Cref{eq:single_qubit_two_desing}
 \item Pauli mixing ensemble from \Cref{eq:single_qubit_pauli_mixing}
 \item Trivial ensemble $\{(1, I)\}$, i.e., no twirling
\end{enumerate}
For the ensemble $\mathcal{V}$, applied to the measure-and-prepare circuit $\mathcal{M}$, we consistently employ the set of unitary operators specified in \Cref{eq:single_qubit_mub_transformations} across all experiments.
To set the QPD coefficents, the QPD from \Cref{theorem_arbitrary_cut} is formulated as $\mathcal{I}^{\otimes n} = c\mathbb{E}_{\mathcal{E}}(\mathcal{C}) + \left(1-c\right)\mathbb{E}_{\mathcal{V}}(\mathcal{M})$, where the Parameter $c$ is determined from the simulated or measured entanglement fidelity $F(\mathcal{C})$ of the channel $\mathcal{C}$ using the relation $c = F(\mathcal{C})^{-1}$.

\subsection{Numerical simulations}\label{sec:simulations}
This section details the methodology and results of the numerical simulations designed to evaluate our QPD protocols.
A central aim of these simulations is to assess the errors of QPDs under diverse channel conditions. 
Specifically, we investigate the influence of channels with varying entanglement fidelities.
Furthermore, we analyze how different magnitudes of coherent errors affect QPDs that employ Pauli mixing ensembles, particularly as \Cref{eq:pauli_mixing_coherent_errors} demonstrates that greater coherent error magnitudes result in potentially larger errors in these QPDs.
To model this range of conditions, encompassing both varying entanglement fidelities and different coherent error magnitudes, we utilize randomly generated channels. 
These simulations thereby establish an ideal performance benchmark for the QPD protocols, allowing for comparison with subsequent hardware experiments and providing initial insights into objectives 2 and 3 under idealized conditions.

\subsubsection{Method}
We simulate the quantum channel $\mathcal{C}$ for the state transfer between devices using an error model adapted from Combes et al.~\cite{Combes2014}.
This model is particularly suitable as it allows for explicit control over the channel's resulting entanglement fidelity and coherent errors.
It achieves this by combining incoherent noise, e.g., arising from environmentally induced decoherence, and coherent errors, e.g., caused by systematic over- or under-rotations from miscalibrated gates, to capture common experimental imperfections.

The incoherent noise component is modeled as a Pauli channel $\mathcal{P}$, defined by
\begin{align}
\mathcal{P}(\rho) = (1-q)\rho + q\left(p_x X\rho X + p_y Y\rho Y + p_z Z\rho Z\right),
\end{align}
where $q \in[0,1]$ quantifies the total error probability, and $p_x,p_y,p_z\ge 0$ with $p_x + p_y + p_z = 1$ describe the distribution of Pauli errors.

The coherent error component is modeled as a unitary transformation $\mathcal{U}(\rho) = U\rho U^\dagger$, where the unitary operator $U$ represents a general rotation:
\begin{align}
U = \exp\left(\frac{-i\theta\vec{n}\cdot\vec{\sigma}}{2}\right).
\end{align}
Here, $\theta$ is the rotation angle determining the coherent error magnitude, $\vec{n} = (n_x, n_y, n_z)$ is a unit vector specifying the rotation axis on the Bloch sphere, and $\vec{\sigma} = (X, Y, Z)$ is the vector of Pauli matrices.

The complete noisy channel $\mathcal{C}$ applies incoherent errors first, followed by the coherent rotation:
\begin{align}
\mathcal{C} = \mathcal{U} \circ \mathcal{P}.
\end{align}
The channel's entanglement fidelity is given by
\begin{align}\label{eq:ent_fidelity_error_model}
F(\mathcal{C}) = (1-q)c^2 + q\left(p_x n_x^2 + p_y n_y^2 + p_z n_z^2\right)s^2,
\end{align}
where $c = \cos(\theta/2)$ and $s = \sin(\theta/2)$~\cite{Combes2014}.

Using this error model, we generate random instances of the channel $\mathcal{C}$ that target specific values for entanglement fidelity $F(\mathcal{C})$ and coherent error rotation angle $\theta$. 
This allows us to study the effect of $F(\mathcal{C})$  on the sampling overhead, and the impact of coherent errors quantified by $\theta$ on the Pauli mixing ensemble.
The generation procedure is as follows: first, a Pauli error distribution $p_x, p_y, p_z$ is randomly selected.
Second, a rotation axis $\vec{n}$ is chosen by sampling uniformly from the surface of the unit sphere. 
Finally, given these selections ($p_x, p_y, p_z$, and $\vec{n}$) and the target values of $F(\mathcal{C})$ and $\theta$, the total Pauli error probability $q$ is calculated by solving \Cref{eq:ent_fidelity_error_model} for $q$.

To evaluate the error scaling of different QPDs, corresponding to the three choices for ensemble $\mathcal{E}$, we simulated the transfer of $500$ Haar-random initial states using various noisy quantum channels $\mathcal{C}$.
Specifically, for each combination of target entanglement fidelity $F(\mathcal{C})$,  chosen from $ \{0.55, 0.7, 0.9\}$  to represent low, medium, and high fidelity scenarios respectively, and coherent rotation angle $\theta$,  selected from $\{0, 0.15, 0.3\}$ to model progressively increasing magnitudes of coherent error, one such random channel instance was generated using the method described.

\begin{figure*}[t]
 \centering
 \includegraphics[trim={0 0.35cm 0 0.25cm},clip]{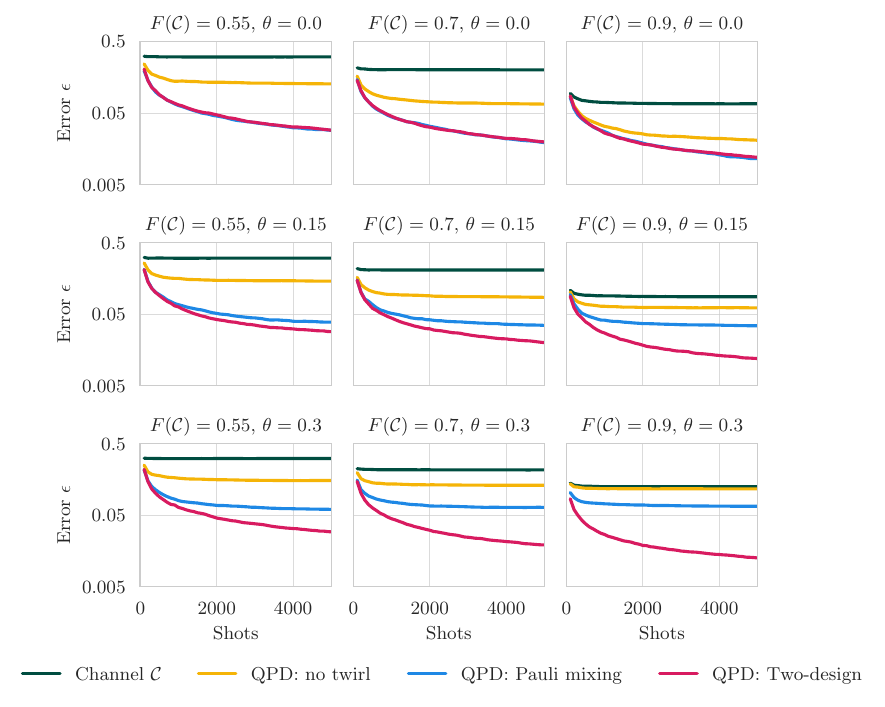}
 \caption{Simulation results: Error scaling of the channel $\mathcal{C}$ and various QPDs applied to channel $\mathcal{C}$, evaluated for different entanglement fidelities $F(\mathcal{C})$ and coherent rotation angles $\theta$.}
 \label{fig:simulation_plot}
\end{figure*}

\subsubsection{Results}
\Cref{fig:simulation_plot} presents the simulation results comparing the performance of different QPDs under various noise conditions.
Each subplot displays the error in the expectation value $\epsilon(N)$ averaged over all $500$ input states and the two observables as a function of the number of shots $N$ for quantum channel $\mathcal{C}$ characterized by a fixed entanglement fidelity $F(\mathcal{C})$ and coherent rotation angle $\theta$.
The logarithmic y-axis in the plots facilitates a clear visualization of the error differences across multiple orders of magnitude. 

Before analyzing the specific performance of each QPD variant, it is helpful to outline how distinct error contributions as introduced in \Cref{sec:error_analysis} manifest in such plots. 
Typically, statistical sampling error $\epsilon_{\text{sampling}}$ is indicated by a decreasing trend in $\epsilon(N)$ as the number of shots $N$ increases, as more data reduces statistical fluctuations. 
In contrast, a systematic error $\epsilon_{\text{bias}}$ becomes apparent if $\epsilon(N)$ converges to a persistent, non-zero value, referred to as an error plateau, even for a large number of shots $N$. 
This plateau signifies a baseline error inherent to the method or its imperfect implementation, which cannot be eliminated by merely increasing $N$.

As illustrated in \Cref{fig:simulation_plot}, the QPD employing a unitary two-design (red line) consistently achieves the lowest average estimation errors across all tested configurations of $F(\mathcal{C})$, $\theta$, and $N$. 
This superior performance is attributed to its ability to transform the general channel $\mathcal{C}$ into the exact depolarizing channel $\mathcal{D}_{F(\mathcal{C})}$ required by the QPD formulation in \Cref{theorem_arbitrary_cut}.
Consequently, this QPD simulates an ideal, noiseless state transfer, thereby avoiding systematic error ($\epsilon_{\text{bias}}=0$) from the QPD itself. 
The remaining error is therefore due to statistical sampling noise.
This behavior is evident in \Cref{fig:simulation_plot}, where the estimation error for the two-design QPD decreases monotonically with increasing $N$, following the standard scaling of $\mathcal{O}(1/\sqrt{N})$, showing no saturation to an error plateau within the range of $N$ depicted.
Furthermore, the estimation error is independent of the coherent error angle $\theta$, as the two-design effectively handles coherent errors. 
Additionally, for a fixed number of shots $N$, the estimation error decreases as the channel’s entanglement fidelity $F(\mathcal{C})$ increases, reflecting the reduced sampling overhead associated with higher-fidelity channels, a key theoretical validation for objective 2.
\looseness=-1

In contrast, the QPD based on the Pauli mixing ensemble (blue line) performs comparably to the unitary two-design QPD only in the absence of coherent errors ($\theta = 0$). 
In this specific scenario, its error scaling with increasing shot number $N$ or higher entanglement fidelity $F(\mathcal{C})$ is equal to that of the two-design QPD. 
However, as \Cref{fig:simulation_plot} illustrates, the presence of coherent errors ($\theta > 0$) significantly degrades the Pauli mixing QPD's performance relative to the two-design. 
This degradation occurs because Pauli mixing fails to handle coherent errors. 
Consequently, it introduces a non-zero systematic error, which increases with the coherent error angle $\theta$. 
This directly explores the influence of ensemble choice on systematic error under coherent noise, central to objective 3.
Visually, this systematic error causes the estimation error in \Cref{fig:simulation_plot} to converge to a non-zero plateau for large $N$, with the plateau height increasing with $\theta$.
Notably, the plots also show that for higher entanglement fidelities $F(\mathcal{C})$, and therefore smaller sampling overhead, the error converges faster towards this error-induced plateau. 
\looseness=-1

The QPD variant that omits twirling for channel $\mathcal{C}$ (yellow line) exhibits the poorest performance among the QPDs.
Since a randomly generated channel $\mathcal{C}$ is generally not depolarizing, this QPD variant suffers from a systematic error even without coherent errors ($\theta = 0$).
Unlike the Pauli mixing case, where the systematic error is primarily driven by $\theta$, the systematic error (and thus the error plateau) for the no-twirl QPD also varies with the entanglement fidelity $F(\mathcal{C})$.

Finally, across all tested configurations, all QPD variations achieve lower estimation errors than directly using the noisy channel $\mathcal{C}$ without the QPD framework (dark green line). 
The error associated with this direct channel use converges most rapidly to a plateau, suggesting it is dominated by the inherent channel imperfections (systematic error) rather than statistical sampling noise.

\begin{mdframed}
 [
 skipabove=15pt,
 innerbottommargin=.3\baselineskip,
 rightmargin=0em,
 leftmargin=0em,
 nobreak=true
 ]
 \textbf{Major observations from simulations:}
 \begin{itemize}[leftmargin=10pt] 
 \item Unitary two-design QPD offers the lowest error, is immune to coherent errors, and its error further decreases with higher entanglement fidelity as a result of smaller sampling overhead.
 \item Pauli mixing QPD performs worse with coherent errors, leading to a coherent-error-dependent plateau; higher entanglement fidelity accelerates convergence to this plateau.
 \item No-twirl QPD is the least effective QPD due to inherent systematic error from the channel's non-depolarizing structure, even without coherent errors.
 \item QPD strategies, especially unitary two-design and Pauli mixing, substantially reduce errors compared to direct use of the noisy channel.
 \end{itemize}
\end{mdframed}

\subsection{Experimental setup for quantum devices}
This section details the methodology for experiments on contemporary quantum devices. 
We begin by explaining the emulation of the quantum channel via teleportation, followed by a description of the quantum devices utilized. 
Subsequently, we outline the experimental workflow, encompassing circuit transpilation, control and characterization of the emulated channel's entanglement fidelity, execution of the circuits for the QPD protocols, and concluding with the postprocessing of the obtained measurement data.

\subsubsection{Channel emulation via teleportation}
To emulate a distributed quantum computation scenario, given the current unavailability of physically separate, interconnected quantum devices, we logically partitioned qubits on a single quantum device into sender and receiver roles.
The quantum channel $\mathcal{C}$ connecting these qubits from the logical partitions was implemented via quantum teleportation (see \Cref{sec:wire_cut_teleportation}).
In practice, noise affects both the preparation of the entangled resource state, resulting in a NME state $\rho$, and the execution of the teleportation operations, resulting in $\tilde{\mathcal{T}}^{\rho}$ instead of $\mathcal{T}^{\rho}$. 
Consequently, the implemented channel is a noisy teleportation process $\mathcal{C}=\tilde{\mathcal{T}}^{\rho}$. 
This setup enables interpretation of our results within the context of wire cutting using a pre-shared NME state $\rho$.

\begin{table}[t]
 \centering
 \begin{tabular}{c|c|c}
 Device & Qubits & Processor type\\
 \hline
 \hline
 \texttt{ibm\_fez} & 156 & Heron r2\\
 \hline
 \texttt{ibm\_torino} & 133 & Heron r1\\
 \hline
 \texttt{ibm\_bruessels} & 127 & Eagle r3
 \end{tabular}
 \caption{Overview of the quantum devices used for the experiments.}
 \label{tab:qpus}
\end{table}

\subsubsection{Devices}
Experiments were performed on IBM superconducting quantum devices~\cite{ibmq} detailed in \Cref{tab:qpus}.
These devices utilize different processor architectures, namely the older Eagle generation and the newer Heron generation. 
The Heron architecture features tunable couplers, yielding significantly lower two-qubit gate error rates than Eagle devices~\cite{AbuGhanem2025}.
Among the Heron devices employed, Heron r2 is a subsequent revision of Heron r1, mainly differing in its increased qubit count. 
Furthermore, the Heron and Eagle architectures also differ in their native two-qubit gates, with Eagle using echoed cross-resonance gates and Heron implementing controlled-Z gates.
Crucially, all utilized devices support dynamic circuits, enabling the mid-circuit measurements and conditional operations required for implementing the classical communication in quantum teleportation.

\subsubsection{Circuit transpilation}
Executing the QPD circuits of \Cref{theorem_arbitrary_cut} requires mapping the logical circuit structure onto the physical qubit topology of the devices~\cite{Leymann2020a}.
We utilized Qiskit's \texttt{VF2Layout} transpiler pass for this task~\cite{VF2Layout}. 
This pass selects suitable physical qubits by matching the circuit's interaction graph with the device's coupling map, considering the device's current error model to minimize potential noise impact.
We used the interaction graph of a standard teleportation circuit (see \Cref{fig:teleportation_circuit}), including preparation of a maximally entangled resource state $\ket{\Phi_2}$. 
This yields a layout of three interconnected qubits: one qubit holding the sender's initial state to be transmitted, an ancilla qubit for the sender's part of the resource state, and the receiver's qubit, which initially holds the receiver's part of the entangled state and subsequently stores the transmitted state after teleportation. 
This three-qubit mapping was consistently applied to all teleportation-based circuits in the QPD. 
For the simpler measure-and-prepare circuits, which only require two qubits (sender state and receiver), the sender's ancilla qubit from this layout was omitted.

\begin{figure*}[t]
 \centering
 \includegraphics[trim={0 0.25cm 0 0.25cm},clip]{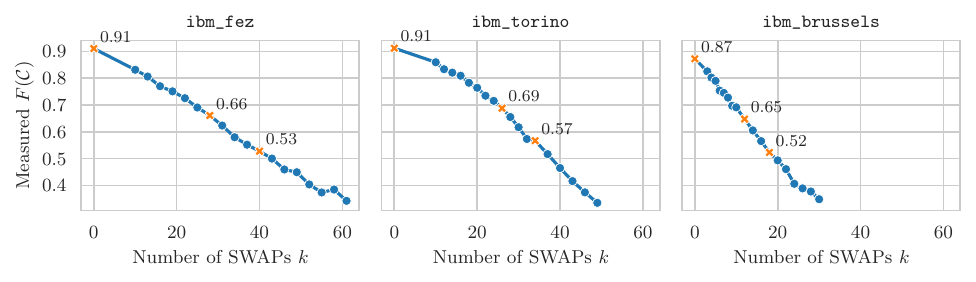}
 \caption{Entanglement fidelity of the teleportation channel as a function of SWAP operations $k$ applied to the resource state, measured using the method described in \Cref{sec:comp_ent_fid}. Highlighted points indicate the configurations used in the subsequent experiments.}
 \label{fig:fidelity_plot}
\end{figure*}

\subsubsection{Controlling and measuring entanglement fidelity}
To compare QPD performance across channels of varying entanglement fidelities on the same physical qubits, addressing objective 2, we augmented inherent device noise with controlled additional noise, thereby enabling systematic tuning of the effective entanglement fidelity for comparable experimental conditions.
The baseline noise from imperfect preparation of the entangled state ($\tilde{\Phi}_2$ instead of $\Phi_2$) and errors during the teleportation operations ($\tilde{\mathcal{T}}$ instead of $\mathcal{T}$) establishes a maximum achievable entanglement fidelity $F(\tilde{\mathcal{T}}^{\tilde{\Phi}_2})<1$ for the used qubits on the given device.

To systematically lower this fidelity, we degraded the entangled resource state $\tilde{\Phi}_2$ by applying $k$ SWAP operations $\tilde{\mathcal{S}}$,  which are inherently noisy due to imperfections in current quantum hardware,  between its constituent qubits before teleportation, swapping its qubit states back and forth. Crucially, it was ensured that these SWAP operations were explicitly executed as intended and not removed by the transpiler. 
The resulting NME state $\rho_k$ is given as:
\begin{align}
\rho_k = \tilde{\mathcal{S}}^k(\tilde{\Phi}_2) = \underbrace{\tilde{\mathcal{S}} \circ \ldots \circ \tilde{\mathcal{S}}}_{k \text{ times}}(\tilde{\Phi}_2).
\end{align}
While ideal SWAP operations $\mathcal{S}$ preserve the maximally entangled state $\mathcal{S}(\Phi_2) = \Phi_2$, noisy SWAPs introduce errors. 
Assuming approximately constant noise per SWAP, increasing $k$ monotonically reduces the entanglement fidelity $F(\tilde{\mathcal{T}}^{\rho_k})$ of the resulting teleportation channel, thereby simulating the use of NME resource states of varying quality.
Crucially, for the two-design QPD, its performance is dictated solely by the entanglement fidelity $F(\tilde{\mathcal{T}}^{\rho_k})$, rendering the detailed structure of the errors in the teleportation channel $\tilde{\mathcal{T}}^{\rho_k}$ and hence in the resource state $\rho_k$ unimportant beyond their collective impact on this fidelity. 
The effectiveness of the Pauli mixing ensemble, in contrast, is more constrained: it is sensitive to specific error characteristics of $\tilde{\mathcal{T}}^{\rho_k}$, particularly to any coherent errors introduced by the teleportation operations. 

To determine suitable values of $k$ for the main QPD experiments and to characterize the device-specific impact of SWAP errors, we first performed calibration runs, measuring the entanglement fidelity $F(\tilde{\mathcal{T}}^{\rho_k})$ for increasing $k$ using the method from \Cref{sec:comp_ent_fid}.
Based on these results, we select three distinct $k$ values corresponding to three different noise levels.
The results of these entanglement fidelity calibration runs are presented in \Cref{fig:fidelity_plot}.

The baseline entanglement fidelities ($k=0$), shown in \Cref{fig:fidelity_plot}, reflect the inherent noise levels of the teleportation process on each device. 
We measured baseline fidelities of approximately $0.91$ for both \texttt{ibm\_fez} and \texttt{ibm\_torino}, and $0.87$ for \texttt{ibm\_bruessels}.
As expected, applying successive SWAP operations leads to a decrease in entanglement fidelity for all devices, although the rate of degradation varies. 
\texttt{ibm\_bruessels} showed the most rapid degradation, followed by \texttt{ibm\_torino}, whereas \texttt{ibm\_fez} demonstrated the slowest decline, suggesting greater resilience to the accumulated errors from these SWAP operations.

Notably, the number of SWAP operations required for the fidelity to drop below the critical $0.5$ threshold, where quantum advantage over classical communication is lost (see \Cref{eq:advantage_qpd_over_wire_cut}), varies significantly across the devices. 
As seen in \Cref{fig:fidelity_plot}, this threshold is crossed after approximately $k=19$ for \texttt{ibm\_bruessels}, $k=37$ for \texttt{ibm\_torino}, and $k=43$ for \texttt{ibm\_fez}.
These results indicate that SWAP operations on Heron-based devices (\texttt{ibm\_fez}, \texttt{ibm\_torino}) exhibit lower accumulated error rates, which in turn leads to higher entanglement fidelity for the teleportation channel compared to the Eagle-based device (\texttt{ibm\_bruessels}).

Based on these calibration findings, for the following QPD experiments, we selected two additional noise levels for each device, representing medium and low entanglement fidelities, beyond the high-fidelity baseline ($k=0$).
All selected fidelites are highlighted in \Cref{fig:fidelity_plot}.
The different numbers of SWAP gates $k$ for the additional noise levels were chosen to satisfy the following criteria.
The primary criterion was to ensure that the sampling overhead factor $\kappa = 2F(\tilde{\mathcal{T}}^{\rho_k})^{-1} -1$ was approximately evenly distributed.
Since $\kappa$ is a non-linear function of the fidelity, this approach naturally results in fidelity values that are not themselves evenly spaced. 
As a secondary criterion, we selected $k$ for each device to ensure the resulting measured fidelities for a given noise level were approximately matched across the different devices to simplify comparison.
Finally, we ensured all selected fidelities remained above the classical threshold of $F(\tilde{\mathcal{T}}^{\rho_k})>0.5$, as values below this offer no advantage in sampling overhead compared to wire cutting that uses only classical communication (see \Cref{eq:advantage_qpd_over_wire_cut})
This procedure resulted in the following additional selected SWAP counts for $k>0$ and their corresponding measured entanglement fidelities:
\begin{itemize}
\item \texttt{ibm\_fez}: $0.66$ ($k=28$) and $0.53$ ($k=40$)
\item \texttt{ibm\_torino}: $0.69$ ($k=26$) and $0.57$ ($k=34$)
\item \texttt{ibm\_bruessels}: $0.65$ ($k=12$) and $0.52$ ($k=18$)
\end{itemize}

\subsubsection{Execution of the QPD}
For efficient comparison across QPD protocols (using unitary two-design, Pauli mixing, or trivial ensembles for channel twirling) and the baseline case of teleportation without QPD, a core set of $15$ unique circuits is executed once for each combination of noise level, one of $25$ Haar-random initial states, and observables $X$ and $Z$.
Critically, every single-shot measurement outcome from these executions is stored for comprehensive postprocessing and data reuse.
This core set of 15 circuits is dictated by the requirements of the most general two-design QPD, which includes $12$ circuits for the twirled teleportation channel using the single-qubit unitary two-design from \Cref{eq:single_qubit_two_desing} and $3$ measure-and-prepare circuits, employing the unitary operators from \Cref{eq:single_qubit_mub_transformations}. 
Importantly, results for all other scenarios are derived by using subsets of these $15$ circuit executions.

This strategy allows analyzing an arbitrary total shot count $N$ during post-processing by utilizing the required number of stored single-shot samples for each relevant circuit. 
Moreover, by deriving results for all compared scenarios from the identical underlying raw measurement data, we significantly minimized statistical fluctuations between scenarios compared to performing independent experimental runs for each. 
Consequently, observed performance differences between scenarios primarily reflect their inherent distinctions, rather than statistical variations from independent experimental runs.

\begin{figure*}[t]
 \centering
 \includegraphics[trim={0 0.35cm 0 0.25cm},clip]{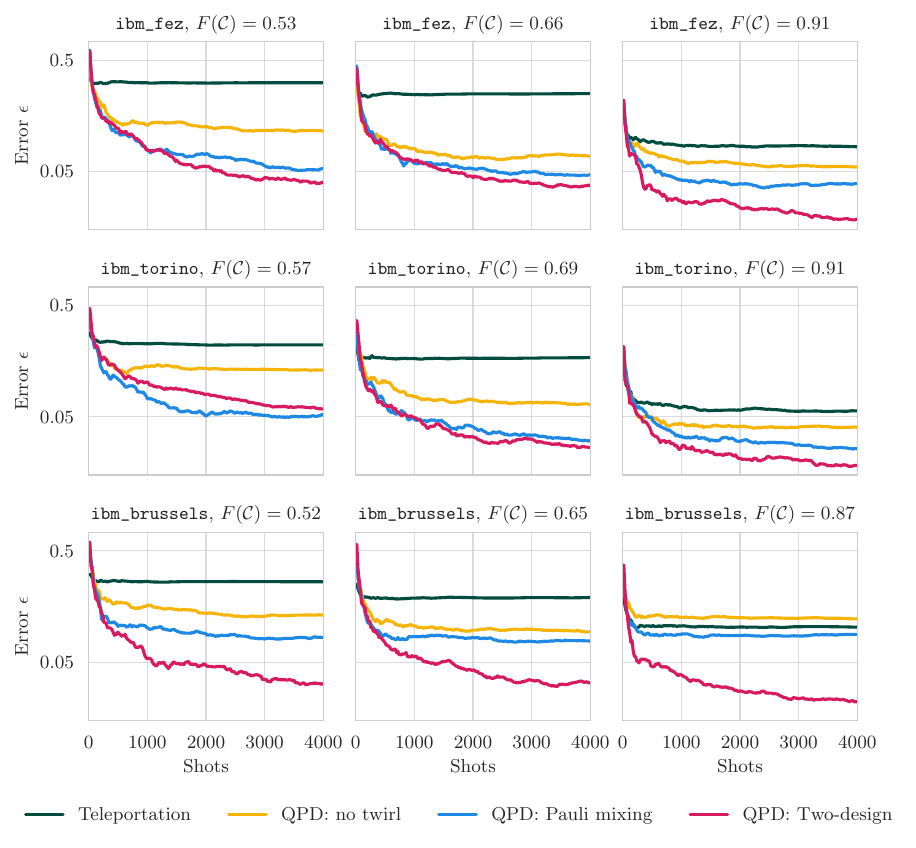}
 \caption{Error scaling for the QPDs and direct teleportation under varying entanglement fidelity across different quantum devices.}
 \label{fig:qpu_plot}
\end{figure*}

\subsubsection{Postprocessing}
Expectation values for QPDs are computed using the estimator in \Cref{eq:qpd_estimator} with the stored single-shot measurement samples.
The number of samples drawn from each circuit's stored data is selected proportionally to that circuit's probability $p_i$ within the QPD.
Specifically, when sampling the QPD with $N$ shots, we use $\lfloor p_iN \rfloor$ samples from the circuit corresponding to the operator $\mathcal{F}_i$ in the QPD.
The total number of samples used, $\sum_i \lfloor p_iN \rfloor$, might be slightly less than $N$ due to rounding, but this discrepancy diminishes for larger $N$.

\subsection{Results: Error scaling with the number of shots}\label{sec:error_scaling_qpu}
This section analyzes the scaling of average estimation error $\epsilon(N)$ with the number of measurement shots $N$ to further investigate sampling overhead on quantum devices.
We compare direct quantum teleportation against the various QPD protocols, using the computed coefficients $c_{\text{com}}= F(\mathcal{C})^{-1}$ derived from measured entanglement fidelities of the calibration runs.
The experimental results validate key findings from numerical simulations and are visualized in \Cref{fig:qpu_plot}, where each subplot represents a specific quantum device and selected entanglement fidelity.

A primary observation from \Cref{fig:qpu_plot} is that the QPD protocols generally achieve significantly lower error plateaus than the direct use of the noisy teleportation channel (dark green line), with one exception on \texttt{ibm\_bruessels} detailed below.
The error associated with this direct channel use quickly saturates at a high plateau, limited by the inherent systematic noise of the used channel. 
In contrast, all QPD protocols demonstrate a more substantial reduction in the error as the total number of shots increases, typically achieving significantly lower error levels.
Initially, QPD errors decrease following the expected $\mathcal{O}(1/\sqrt{N})$ scaling before potentially saturating at these lower levels. 
This confirms that sampling error dominates QPD performance at lower shot counts and that the QPDs substantially reduces the error compared to direct channel use.

The unitary two-design (red line) yields the lowest errors across most configurations, demonstrating its robustness against arbitrary noise experimentally.
The single observed instance where the Pauli mixing QPD (blue line) appeared to perform marginally better is likely a statistical artifact, as these experimental results using the quantum devices are averaged over only $25$ random initial states for the two observables $X$ and $Z$, compared to $500$ initial states in the simulation.
Moreover, the error for the two-design QPD continuously decreases with an increasing number of shots $N$ in these experiments, showing no signs of a limiting error plateau within the tested range. 
Furthermore, it clearly exhibits the expected trend: higher entanglement fidelity $F(\mathcal{C})$ of the channel $\mathcal{C}$ results in lower errors for a fixed $N$, confirming that reduced sampling overhead improves accuracy.

The Pauli mixing QPD (blue line) serves as an experimental probe for coherent errors, with observed results aligning with simulations for non-zero coherent errors from \Cref{sec:simulations}.
Its performance relative to the two-design highlights hardware differences.
On Heron devices (\texttt{ibm\_fez}, \texttt{ibm\_torino}), a small error gap between the two QPDs suggests minimal coherent errors in the teleportation implementation, enabling the Pauli mixing QPD, which only handles Pauli errors, to perform well. 
Conversely, on the Eagle device \texttt{ibm\_bruessels}, a significantly larger error gap and clear convergence to an error plateau indicate stronger coherent errors.
These coherent errors are not effectively handled by the Pauli mixing QPD, creating a persistent systematic error plateau that limits performance improvement, even when increasing fidelity $F(\mathcal{C})$, as seen by the plateau being largely independent of $F(\mathcal{C})$. 
The faster convergence to this plateau at higher fidelities further illustrates the dominance of this systematic error over the diminishing sampling noise for this method under these conditions.

As expected, omitting the channel twirl (yellow line) yields the largest errors among the QPD methods.
This poor performance is especially highlighted on \texttt{ibm\_bruessels} at $F(\mathcal{C}) = 0.87$, where, with the computed QPD coefficients, the no-twirl QPD performs even worse than the directly using channel $\mathcal{C}$.

\begin{mdframed}
 [
 skipabove=15pt,
 innerbottommargin=.3\baselineskip,
 rightmargin=0em,
 leftmargin=0em,
 nobreak=true
 ]
 \textbf{Major observations from error scaling:}
 \begin{itemize}[leftmargin=10pt] 
 \item Quantum device experiments validate trends predicted by simulations.
 \item Unitary two-design QPD generally achieves the smallest error on all devices, with its accuracy improving with higher entanglement fidelity $F(\mathcal{C})$ due to lower sampling overhead.
 \item Pauli mixing QPD performs well on Heron devices (indicating low coherent noise) but is limited by systematic error plateau on the Eagle device (indicating higher coherent noise).
 \item The effectiveness of Pauli mixing on Heron devices validates its viability for reducing QPD circuits on hardware with low coherent noise.
 \end{itemize}
\end{mdframed}

\begin{figure*}[t]
 \centering
 \includegraphics[trim={1.22cm 0cm 0 0.25cm},clip]{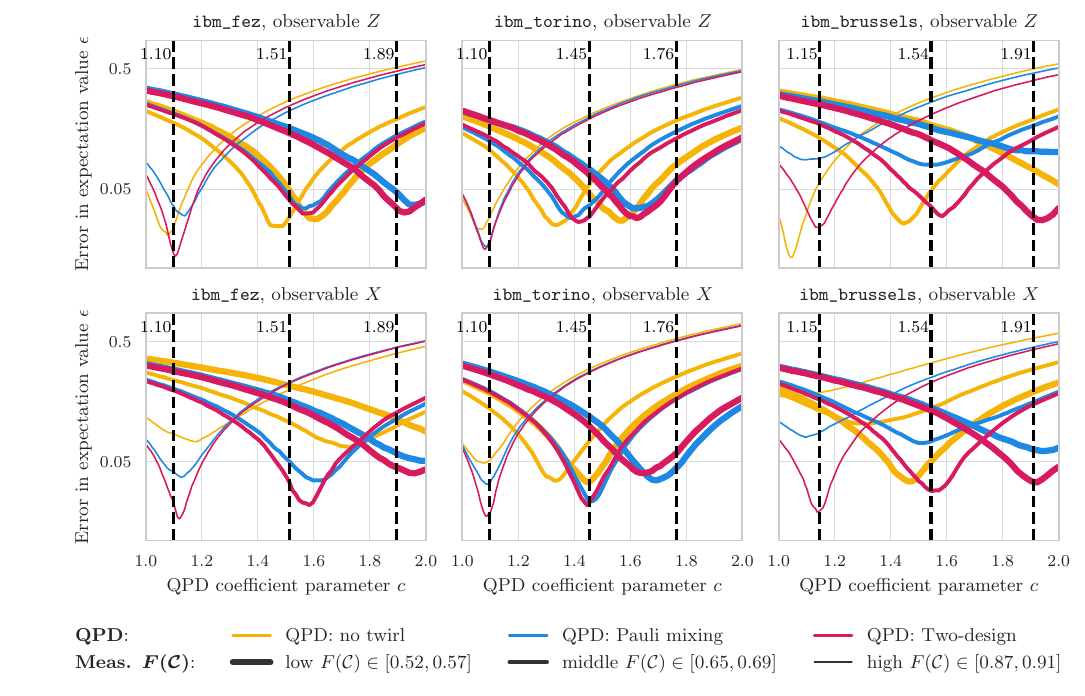}
 \caption{Analysis of the error with QPD coefficients $c$ and $1-c$, where the dotted vertical lines indicate the experimentally computed QPD coefficients $c_{\text{com}}$ based on the different measured entanglement fidelities $F(\mathcal{C})$.}
 \label{fig:fidelity_analysis}
\end{figure*}

\subsection{Results: Minimal error and validation of QPD coefficients} 
This section aims to further validate the method in \Cref{sec:comp_ent_fid} for determining QPD coefficients from measured entanglement fidelity and in doing so, to demonstrate the practical feasibility of our approach (objective 1). 
The core of this validation involves analyzing the impact of deviations from the QPD coefficients specifically computed using this method. 
To conduct this analysis, we reprocess stored single-shot measurement data, calculating expectation values across a range of hypothetical QPD coefficients, which are parameterized as $c$ and $1-c$. 
By scanning through possible values for $c$, we identify the empirical value $c_{\text{opt}}$ that minimizes the average estimation error $\epsilon(N)$. 
Comparing $c_{\text{opt}}$ to $c_{\text{com}}= F(\mathcal{C})^{-1}$, which is derived from the measured channel fidelity $F(\mathcal{C})$, validates our fidelity measurement and coefficient calculation method.
Recall that $c$ relates to the sampling overhead $\kappa = 2c-1$, and the constraint $F(\mathcal{C}) \in [0.5, 1]$ implies $c \in [1, 2]$. 
\Cref{fig:fidelity_analysis} illustrates this analysis, plotting the average estimation error, using $4000$ shots for observables $X$ and $Z$, as a function of the hypothetical QPD coefficient parameter~$c$. 
Dashed vertical lines indicate the calculated $c_{\text{com}}$ values.

\subsubsection{Two-design and Pauli-mixing QPDs}\label{sec:qpd_coefficient_analysis}
For the QPD protocols that employ channel twirling using either the unitary two-design (red line) or Pauli mixing ensembles (blue line), the results presented in \Cref{fig:fidelity_analysis} reveal several key findings.
First, the empirical optimum $c_{\text{opt}}$ is largely independent of the measured observable ($X$ or $Z$) and very similar for both ensembles. 
This empirical optimum $c_{\text{opt}}$, which corresponds to the coefficient value that minimizes the error shown in \Cref{fig:fidelity_analysis}, aligns well with the coefficient $c_{\text{com}}$ computed via the method from \Cref{sec:comp_ent_fid} (indicated by dashed lines), particularly at high entanglement fidelities.
The observed discrepancies between the empirically optimal $c_{\text{opt}}$ and the computed $c_{\text{com}}$ at lower fidelities where $c_{\text{opt}}$ values are correspondingly larger, likely stem from the error magnification inherent in the calculation $c_{\text{com}}= F(\mathcal{C})^{-1}$.
This inverse relationship makes $c_{\text{com}}$ increasingly sensitive to absolute errors in the measured $F(\mathcal{C})$ as $F(\mathcal{C})$ itself decreases.

Second, the minimum achievable error at $c_{\text{opt}}$ shows different dependencies on entanglement fidelity $F(\mathcal{C})$ for the two ensembles. 
For the unitary two-design QPD, higher entanglement fidelity generally yields lower minimum errors, as exhibited in \Cref{fig:qpu_plot} by the lower minima of the error curves corresponding to higher entanglement fidelity $F(\mathcal{C})$.
This confirms that a reduced sampling overhead $\kappa = 2c_{\text{opt}}-1$ improves accuracy for a fixed shot count $N$ and directly addresses objective 2.
In contrast, the minimum achievable errors for the Pauli mixing QPD show no clear correlation with the entanglement fidelity $F(\mathcal{C})$, as adjusting the coefficients does not correct unhandled coherent errors for the Pauli mixing QPD.
This finding is consistent with the error plateaus observed in error-scaling experiments, both in simulations (\Cref{fig:simulation_plot}) and on quantum devices (\Cref{fig:qpu_plot}).

Third, as anticipated from simulations, unitary two-design QPDs generally achieve lower errors than Pauli mixing QPDs.
This suggests coherent errors in the teleportation channel, which the two-design can handle, but Pauli mixing does not.
This performance gap is notably smaller for Heron-based devices (\texttt{ibm\_fez}, \texttt{ibm\_torino}) than for the Eagle-based \texttt{ibm\_bruessels}. 
This implies stronger coherent errors in the Eagle device's channel implementation, which more significantly degrade Pauli mixing QPD performance relative to the two-design.

\subsubsection{QPD without channel twirl}
In contrast, the QPD variant that omits channel twirling (yellow line), i.e., uses the trivial twirl, exhibits markedly different behavior in \Cref{fig:fidelity_analysis}.
First, its empirically optimal coefficient $c_{\text{opt}}$ (corresponding to the minima of the yellow error curves) often deviates significantly from the computed value $c_{\text{com}}$ (indicated by vertical dashed lines), which is derived from the entanglement fidelity.
Second, both the optimal coefficient $c_{\text{opt}}$ and the minimum achievable error value show strong dependence on the estimated observable ($X$ or $Z$). 
For example, on \texttt{ibm\_fez}, the non-twirled QPD achieves low errors for observable $Z$ (comparable to twirled QPDs, though at a different $c_{\text{opt}}$), but exhibits significantly higher minimum errors for observable $X$, again with a different optimal coefficient.

This behavior indicates that the implemented noisy teleportation channel $\mathcal{C}=\tilde{\mathcal{T}}^{\rho_k}$ is not inherently depolarizing.
As established in \Cref{sec:main_qpd}, if channel $\mathcal{C}$ were genuinely depolarizing, specific QPD coefficients would exist that enable an accurate simulation of the identity channel, and crucially, this simulation's effectiveness, including optimal coefficients and minimum error, would be independent of the observable being measured.
However, when the channel twirl is omitted, $\mathcal{C}$ is not transformed into the requisite depolarizing form, thereby violating a core assumption of \Cref{theorem_arbitrary_cut}.
As a result, the observed optimal coefficients are no longer expected to align with $ F(\mathcal{C})^{-1}$.
Furthermore, without the averaging effect of channel twirling, the performance of the QPD becomes sensitive to the specific interplay between the channel's non-depolarizing structure and the chosen measurement observable, explaining the observed dependencies.
\looseness=-1

\begin{mdframed}
 [
 skipabove=15pt,
 innerbottommargin=.3\baselineskip,
 rightmargin=0em,
 leftmargin=0em,
 nobreak=true
 ]
 \textbf{Major observations from QPD coefficients:}
 \begin{itemize}[leftmargin=10pt] 
 \item For two-design and Pauli mixing QPDs, computed coefficients align well with empirical optima (especially at high $F(\mathcal{C})$), validating the calibration method of \Cref{sec:comp_ent_fid}.
 \item Unitary two-design QPD error improves with higher entanglement fidelity $F(\mathcal{C})$; Pauli mixing QPD is limited by coherent errors, consistent with simulation findings.
 \item Without channel twirling, optimal QPD coefficients deviate significantly from computed coefficients, and performance becomes strongly observable-dependent, confirming the channel is not inherently depolarizing.
 \end{itemize}
\end{mdframed}

 \section{Discussion}\label{sec:discussion}

\Cref{theorem_arbitrary_cut} introduces a general QPD, enabling computations that leverage noisy quantum channels for state transfer between distributed devices. 
Our experiments, detailed in \Cref{sec:experiments}, confirm the feasibility of implementing this QPD on a single quantum device and demonstrate its higher accuracy over direct noisy channel use. 
This section further explores the properties of the proposed QPD and addresses the limitations of our experimental validation.

\subsection{Analysis and comparison of the proposed QPD}

A key advantage of the QPD derived in \Cref{theorem_arbitrary_cut} is its simple calibration and applicability to arbitrary channels $\mathcal{C}$.
The QPD's coefficients depend solely on the channel's entanglement fidelity $F(\mathcal{C})$. 
Determining this single parameter is sufficient for calibration, and as demonstrated experimentally, the entanglement fidelity $F(\mathcal{C})$ can be measured efficiently between distributed devices using the method described in \Cref{sec:comp_ent_fid}.
Furthermore, the initial QPD calibration remains valid and can be reused provided the channel $\mathcal{C}$'s noise characteristics are stable. 
Should these characteristics drift, the QPD coefficients can be periodically updated based on the current entanglement fidelity $F(\mathcal{C})$, similar to recalibrations that maintain the operational accuracy of individual quantum devices~\cite{Deng2024}.

From a practical perspective, minimizing the number of distinct circuits within a QPD is crucial.
This reduction directly lowers the costs associated with circuit transpilation and limits the number of unique circuits from which samples must be drawn during execution.
\Cref{lemma_mp_channel} guarantees that the number of circuits simulating the zero-fidelity depolarizing channel $\mathcal{D}_0$ component is minimal. 
However, the overall QPD construction from \Cref{theorem_arbitrary_cut} does not necessarily minimize the total circuit count for simulating the identity channel via the channel $\mathcal{C}$.
This potential sub-optimality arises from two main factors.
First, the minimality of the overall construction depends on the structural properties of $\mathcal{C}$, which dictate the complexity of the required channel twirl. 
For instance, implementing the twirl for an arbitrary channel requires a unitary two-design, whereas a smaller Pauli mixing ensemble suffices for a Pauli channel.
Second, alternative QPD constructions, differing from the depolarizing-channel-based formulation in \Cref{eq:qpd_identiy_depol}, might exist that achieve lower total circuit counts for specific channels.

Although such alternative QPDs could potentially reduce circuit numbers, they likely involve trade-offs in generality or ease of use.
For example, the specialized QPD developed for teleportation channels using pure NME resource states requires only $2^n+1$ circuits for an $n$-qubit state transfer~\cite{Bechtold2025}.
This count is significantly lower than for our general construction.
Our QPD uses $|\mathcal{E}| + 2^n+1$ circuits in total: $2^n+1$ circuits for the $\mathcal{D}_0$ component alone, with an additional $|\mathcal{E}|$ circuits  generated by twirling channel $\mathcal{C}$ using ensemble $\mathcal{E}$. 
The magnitude of this $|\mathcal{E}|$ term can be considerable, especially if $\mathcal{E}$ must form a unitary two-design (see \Cref{eq:size_two_design}).
However, this efficiency of the specialized QPD comes with significant drawbacks.
Firstly, it applies strictly to teleportation implemented with pure NME states, whereas \Cref{theorem_arbitrary_cut} handles arbitrary channels $\mathcal{C}$.
Secondly, this specialized QPD requires determining all $2^n$ Schmidt coefficients of the NME state as input parameters~\cite{Bechtold2025}, significantly increasing calibration complexity compared to our single-parameter approach. 
Thirdly, it implicitly assumes noiseless teleportation operations, as calibration relies only on the pure resource state properties, neglecting potential noise within the process itself.
This example for teleportation with pure NME states illustrates that QPDs optimized for specific channel structures to minimize circuit count may demand more complex calibration. 
Identifying the specific structure and determining numerous parameters may even necessitate full characterization methods like quantum process tomography~\cite{Mohseni2008} between devices, which in turn enables the construction of QPDs that are highly tailored to the channel and can potentially yield significant benefits.
In contrast, our approach maintains generality and simplicity by being calibrated solely by the entanglement fidelity.

Beyond minimizing the number of distinct circuits, transpilation costs of a QPD can also be reduced by employing parameterized circuit templates~\cite{CarreraVazquez2024}.
With such templates, different operators within the QPD can ideally be realized by merely adjusting parameters, thus requiring only a single template circuit to be transpiled. 
This use of parameterization is also relevant to our work, particularly for implementing the channel twirls required by \Cref{theorem_arbitrary_cut}. 
Indeed, existing constructions, such as the approximate unitary two-design by Nakata et al.~\cite{Nakata2017}, already feature an efficient, single parameterized template circuit.
While a detailed implementation using parameterized ensembles was beyond the scope of this particular study, future experiments employing our QPD could effectively leverage them for the required twirls. 
Adopting this approach would mean that only circuit parameters require adjustment, thereby reducing transpilation to a single instance per ensemble and offering the potential for substantial savings in pre-processing time.

\subsection{Experimental scope and limitations}\label{sec:experimental_limitations}
Our experiments provide valuable proof-of-principle demonstrations, but several limitations should be acknowledged.
A primary limitation is that the experiments were confined to single-qubit state transfers.
For these proof-of-principle experiments, we focused on this simpler scenario to affirm our method's core viability.
Consequently, we cannot draw direct conclusions about the performance of joint multi-qubit state transfers using a single QPD instance, although \Cref{theorem_arbitrary_cut} theoretically supports this approach.
An alternative for implementing multi-qubit state transfer involves using multiple single-qubit QPDs in parallel. 
This parallel method might simplify the construction of twirling ensembles (e.g., via product structures) but could potentially incur higher sampling overhead compared to joint transfers~\cite{Bechtold2025}. 
The performance trade-offs associated with these different multi-qubit implementation strategies were not experimentally evaluated in this work.

Furthermore, the QPD protocols were investigated in an isolated context, separate from their integration within larger quantum algorithms.
As a result, this study does not address the potential complexities or performance implications arising in broader practical computations.
However, studying the QPDs in this controlled setting is crucial for a precise characterization of their core functionality and response to noise, thereby establishing the necessary baseline for future investigations within specific application contexts.

A further experimental limitation arises from the specific implementation of the quantum channel. 
In our hardware experiments, the channel was realized exclusively through quantum teleportation, with noise introduced by applying SWAP operations to the entangled resource state.
Although this specific implementation might narrow the demonstrated generality across all conceivable channel types, two key observations mitigate this concern.
First, our hardware results align well with the simulations incorporating a more general noise model as outlined in \Cref{sec:simulations}.
Second, the procedure outlined in \Cref{sec:wire_cut_teleportation} shows that an arbitrary channel can, in principle, be converted into such a teleportation channel, implying that this implementation represents a relevant and generalizable scenario.
Furthermore, the performance of the QPD variant utilizing a unitary two-design is theoretically expected to depend solely on the channel's entanglement fidelity, irrespective of its specific noise structure of the channel.

Regarding the experimental setting, all hardware experiments were performed by allocating separate sets of qubits within a single physical quantum device.
This setup inherently bypasses the practical challenges associated with truly distributed quantum computing, such as inter-device synchronization and classical communication latency~\cite{DiAdamo2021}.
Additionally, while the observed intra-device noise is representative for current single-device operations, it might not fully represent the noise characteristics of future networked quantum devices. 
Nonetheless, as previously noted, theory predicts that the performance of the unitary two-design QPD depends solely on the channel's entanglement fidelity.
This suggests the specific origin of the noise, whether intra- or inter-device, is less critical than the overall entanglement fidelity.

Finally, the study relied on a specific set of superconducting quantum devices. 
Although different superconducting hardware architectures were included (IBM's Eagle and Heron processors), variations across distinct hardware platforms necessitate caution when extrapolating quantitative results, due to differing noise characteristics and physical implementation details. 
However, the key qualitative insights derived from our experiments align with theoretical predictions that are independent of specific hardware, suggesting these fundamental findings should remain valid on other quantum platforms capable of executing the required operations.
 \section{Related work}\label{sec:related_work}
Strategies for utilizing imperfect quantum components are critical for advancing quantum computation. 
Our approach to simulating ideal quantum state transfers across noisy interconnects relates to, yet diverges from, established and emerging research areas.

A primary point of comparison is traditional quantum channel distillation~\cite{Liu2020,Regula2021,Kechrimparis2025}. 
While the objective of realizing the behavior of an identity channel using noisy channels is shared, the methodology is fundamentally different. 
Traditional channel distillation aims to physically convert multiple instances of a noisy channel into fewer, higher-fidelity physical channels, often with probabilistic success. 
In contrast, our approach achieves a virtual realization of the ideal identity channel. 
It ensures the correct measurement statistics are reproduced through classical post-processing of data obtained from operations involving the given noisy channel, without physically constructing an improved channel.

Our technique is an instance of virtual resource distillation, a broader concept where QPDs are used to quasiprobabilistically simulate resourceful target objects, e.g., states or channels, using less resourceful ones~\cite{Takagi2024,Yuan2024}.
For quantum states, the feasibility has been experimentally demonstrated~\cite{Zhang2024}.
Related theoretical work on channels examined simulating a general target channel using a given channel $\mathcal{C}$~\cite{Zhao2024a}. 
However, the investigated setting relies on non-signaling resources, e.g., shared randomness, and lacks the classical communication inherent in the measure-and-prepare circuits of \Cref{theorem_arbitrary_cut}.

Our work advances quantum circuit cutting for distributed quantum computing, where large circuits are partitioned via QPDs for execution across smaller devices, typically using only classical communication~\cite{Bechtold2023b}. 
Our QPD enhances wire cuts (simulated state transfers) by using noisy quantum interconnects to reduce sampling overhead.
Pednault~\cite{Pednault2023} uses a similar QPD for wire cutting based on constructing depolarizing channels via two-designs.
However, this method does not incorporate noisy quantum interconnects to reduce sampling overhead, nor does it leverage the minimal measure-and-prepare circuit construction that we detail in \Cref{lemma_mp_channel}.
Distinct from our wire cutting focus, gate cutting simulates the action of a multi-qubit gate via a QPD of local operators~\cite{Piveteau2024,Mitarai2021,Ufrecht2023a,Ufrecht2024}.
Although circuit cutting experiments have successfully distributed computations using only classical communication~\cite{CarreraVazquez2024}, they, unlike our approach, omit leveraging noisy quantum interconnects for sampling overhead reduction.
Moreover, automated tools for identifying optimal cut locations for both wires and gates cut have been developed~\cite{Tang2021,Brandhofer2024,Kan2024}, which will also be valuable for implementations leveraging noisy quantum interconnects.
While circuit cutting primarily targets circuit size reduction, it can also enhance result fidelity~\cite{Perlin2021,Ayral2021,Bechtold2023a}.
However, this must be carefully evaluated in architectures with noisy interconnects, as interconnect quality may introduce additional performance trade-offs.

The underlying quasiprobability framework in our approach is versatile and finds numerous applications.
It forms the basis for quantum error mitigation techniques like probabilistic error cancellation, which simulates noise-free circuit execution using runs on noisy quantum devices~\cite{Temme2017}.
Other diverse applications include simulating non-Clifford channels using only Clifford channels~\cite{Bennink2017} and even simulating unphysical dynamics, such as non-completely positive maps, using physical quantum operations~\cite{Regula2021a}.

Finally, while QPD-based simulations, including our proposed method, offer near-term strategies for leveraging noisy, interconnected quantum devices, large-scale distributed quantum computation ultimately requires fault tolerance via quantum error correction.
Foundational work showed networked topological error correction codes can tolerate highly noisy quantum interconnects, given sufficient fidelity within each individual device~\cite{Nickerson2013}. 
Surface codes, for instance, maintain fault tolerance even when interconnect noise significantly surpasses local operation noise~\cite{Ramette2024}.
Research also explores alternative modular codes, like quantum low-density parity-check codes for specific connectivities~\cite{Strikis2023} and hyperbolic Floquet codes offering efficient encoding and simpler distributed checks~\cite{Sutcliffe2025}. \section{Conclusion}\label{sec:conclusion}
This work demonstrates how noisy quantum interconnects can be practically utilized to quasiprobabilistically simulate high-fidelity state transfers between distributed devices. 
The proposed QPD, detailed in \Cref{theorem_arbitrary_cut}, reduces the sampling overhead for the simulated state transfer according to interconnect's entanglement fidelity and allows simple calibration using the method presented in \Cref{sec:comp_ent_fid}. 
Alongside this, we analyzed strategies to reduce the number of required circuit variants within the QPD and their associated accuracy trade-offs.

Experimental validation on contemporary quantum devices confirmed the feasibility of this proposed QPD, including its calibration.
We successfully demonstrated the predicted reduction in sampling overhead under realistic noise conditions. 
Significantly, the simulated transfer achieved a higher effective fidelity compared to direct state transfer over the same noisy interconnect.
This advantage was maintained even when using approximated QPDs designed to reduce the number of distinct circuit variants.

These findings bridge the gap between traditional wire cutting, which omits quantum interconnects entirely, and the ideal of using error-free quantum interconnects. 
By harnessing noisy interconnects, our QPD framework offers a flexible, near-term strategy for distributed quantum computation where performance scales with improving interconnect quality.

Our experimental limitations, discussed in \Cref{sec:experimental_limitations}, highlight important directions for future research.
These include extending experiments to multi-qubit systems, integrating these simulated transfers within larger distributed algorithms, testing across a broader range of quantum channel implementations, and validating performance in genuine distributed quantum computing architectures.
Furthermore, the QPD framework presented here for state transfer could be extended to directly simulate other distributed operations, e.g., multi-qubit gates, potentially leveraging noisy interconnects to reduce overheads in those settings as well. 
\section*{Acknowledgment}
This work was partially funded by the BMWK projects \textit{EniQmA} (01MQ22007B) and \textit{SeQuenC} (01MQ22009B).
We acknowledge the use of IBM Quantum Credits for this work. The views expressed are those of the authors, and do not reflect the official policy or position of IBM or the IBM Quantum team.

\bibliographystyle{quantum_trunc}
\bibliography{bibliography}

\onecolumn
\appendix
\section{Pauli Operators and fundamental lemmas}\label{sec:appendix_paulis}
This appendix introduces several fundamental properties of Pauli operators that are essential for subsequent proofs in this work. 
Although these properties are well-known, we include them here to ensure completeness and provide a self-contained presentation. 
We first introduce notation and present relevant summation identities. 
Next, we define the Pauli operators $X_{\vec{a}}$ and $Z_{\vec{a}}$, along with products of them. 
Finally, we establish important identities related to Pauli channels, which will be utilized in later proofs.

\subsection{Notation}
For an $n$-qubit system, a computational basis state is represented by a binary vector $\vec{k}=(k_0, \ldots ,k_{n-1})\in \{0,1\}^{n}$ and is constructed as the tensor product  $\ket{\vec{k}} = \bigotimes_{i=0}^{n-1} \ket{k_i}$.
The symbol $\oplus$ denotes addition modulo 2, which is applied bitwise to these binary vectors.
The resulting vector $\vec{k} \oplus \vec{l}$ corresponds to the quantum state $\ket{\vec{k}\oplus \vec{l}}$.
The computational basis states are orthonormal~\cite{Nielsen2009}, such that their inner product is given by the Kronecker delta for vector arguments:
\begin{align}\label{eq:def_kron_delta}
    \braket{\vec{k}|\vec{l}} = \delta_{\vec{k},\vec{l}}=\begin{cases}
        1,& \text{if }\vec{k}=\vec{l}\\
        0,& \text{otherwise}.
    \end{cases}
\end{align}

\subsection{Summation identities}
We start with the following summation identity:
\begin{lemma}\label{lemma_minus_one}
    For $\vec{a}\in \{0,1\}^n$, it holds that
    \begin{align}
        \sum_{\vec{b}\in\{0,1\}^n} (-1)^{\vec{a} \cdot \vec{b}} = 2^n \delta_{\vec{a},\vec{0}}
    \end{align}
    where $\vec{a} \cdot \vec{b}$ is the dot product of the binary vectors $\vec{a}$ and $\vec{b}$ and $\delta_{\vec{a},\vec{0}}$ is the Kronecker delta.
\end{lemma}
\begin{proof}
    We consider two cases.
    When $\vec{a}=\vec{0}$, and utilizing the dot product definition $\vec{a} \cdot \vec{b}=\sum_{i=0}^{n-1} a_i b_i$:
    \begin{align}
        \sum_{\vec{b}\in\{0,1\}^n} (-1)^{\vec{0} \cdot \vec{b}} = \sum_{\vec{b}\in\{0,1\}^n} (-1)^{0} = 2^n =2^n\delta_{\vec{0},\vec{0}}.
    \end{align}
    When $\vec{a}\neq\vec{0}$, the derivation proceeds as follows:
    \begin{align}
        \sum_{\vec{b}\in\{0,1\}^n} (-1)^{\vec{a} \cdot \vec{b}}
        &= \sum_{\vec{b}\in\{0,1\}^n} (-1)^{\sum_{i=0}^{n-1} a_i b_i} \label{eq:proof_step_0} \\
        &= \sum_{b_0,\dots,b_{n-1} \in \{0,1\}} \prod_{i=0}^{n-1} (-1)^{a_i b_i} \label{eq:proof_step_1} \\
        &= \prod_{i=0}^{n-1} \left( \sum_{b_i \in \{0,1\}} (-1)^{a_i b_i} \right) \label{eq:proof_step_2_intermediate} \\
&= \prod_{i=0}^{n-1} \left( 1 + (-1)^{a_i} \right). \label{eq:proof_step_3}
    \end{align}
    Since by assumption $\vec{a}\neq\vec{0}$, there must be at least one index $k \in \{0, \ldots, n-1\}$ for which $a_k=1$. 
    For this specific index $k$, the corresponding factor in the product is $ 1 + (-1)^{a_k}=0$.
    Because at least one factor in the product is zero, the entire product in \Cref{eq:proof_step_3} is zero.
    Thus, we conclude for the case $\vec{a}\neq\vec{0}$ that
    \begin{align}
    \sum_{\vec{b}\in\{0,1\}^n} (-1)^{\vec{a} \cdot \vec{b}} = 0 = \delta_{\vec{a},\vec{0}}
    \end{align}
    Combining these two cases, we obtain the desired result.
\end{proof}

An immediate extension of \Cref{lemma_minus_one} is the following:
\begin{lemma}\label{lemma_double_minus_one}
    For $\vec{a}, \vec{b} \in \{0,1\}^n$:
    \begin{align}
        \sum_{\vec{c}\in\{0,1\}^n} (-1)^{\vec{a} \cdot \vec{c}}(-1)^{\vec{b} \cdot \vec{c}} = 2^n \delta_{\vec{a}, \vec{b}}
    \end{align}
\end{lemma}
\begin{proof}
We proceed as follows:
    \begin{align}
        \sum_{\vec{c}\in\{0,1\}^n} (-1)^{\vec{a} \cdot \vec{c}}(-1)^{\vec{b} \cdot \vec{c}} &= \sum_{\vec{c}\in\{0,1\}^n} (-1)^{\vec{a} \cdot \vec{c} \oplus \vec{b} \cdot \vec{c}}\\
        &= \sum_{\vec{c}\in\{0,1\}^n} (-1)^{(\vec{a} \oplus \vec{b} ) \cdot \vec{c}}\\
        &=2^n \delta_{\vec{a} \oplus \vec{b} , \vec{0}}\label{eq:proof_double_minus_one_1}\\
        &=2^n \delta_{\vec{a}, \vec{b}}\label{eq:proof_double_minus_one_2}
    \end{align}
where we applied \Cref{lemma_minus_one} to obtain \Cref{eq:proof_double_minus_one_1} and \Cref{eq:proof_double_minus_one_2} holds since the condition  $\vec{a} \oplus \vec{b} = \vec{0}$ holds exactly when $\vec{a}=\vec{b}$.
\end{proof}

\subsection{Representations and products of Pauli operators}

The operator $X_{\vec{a}}$ defined in \Cref{eq:X_and_Z} can be expressed as the following:
\begin{align}
        X_{\vec{a}} &= \bigotimes_{i=0}^{n-1} X^{a_i}\\
        &= \bigotimes_{i=0}^{n-1} \left(\ketbra{0}{0 \oplus a_i} + \ketbra{1}{1 \oplus a_i}\right)\\
        &= \bigotimes_{i=0}^{n-1} \sum_{k_i \in \{0, 1\}}\ketbra{k_i}{k_i \oplus a_i}\\
        &= \sum_{k_0 \in \{0,1\}} \ldots \sum_{k_{n-1}\in\{0,1\}} \bigotimes_{i=0}^{n-1} \ketbra{k_i}{k_i\oplus a_i} \label{eq:x_a_rep_intermediate}.
\end{align}
To reach the final form, we use the property of the tensor product to separate the kets and bras: 
\begin{align}
    \bigotimes_{i=0}^{n-1} \ketbra{k_i}{k_i\oplus a_i} &= \left(\bigotimes_{i=0}^{n-1} \ket{k_i}\right) \left(\bigotimes_{i=0}^{n-1}\bra{k_i\oplus a_i} \right).
\end{align}
We then identify the resulting tensor products with their multi-qubit state vector notation, i.e., $\ket{\vec{k}} = \bigotimes_{i=0}^{n-1} \ket{k_i}$ and $\bra{\vec{k} \oplus\vec{a}} = \bigotimes_{i=0}^{n-1}\bra{k_i\oplus a_i}$.
This gives the final compact expression:
\begin{align}
        X_{\vec{a}} &= \sum_{\vec{k} \in \{0,1\}^n} \ketbra{\vec{k}}{\vec{k} \oplus\vec{a}}.\label{eq:x_a_rep}
\end{align}
Moreover, the operator $Z_{\vec{a}}$ from \Cref{eq:X_and_Z} can be expressed as the following,
\begin{align}
        Z_{\vec{a}} &= \bigotimes_{i=0}^{n-1} Z^{a_i}\\
        &= \bigotimes_{i=0}^{n-1} \left(\ketbras{0} + (-1)^{a_i}\ketbras{1}\right)\\
        &= \bigotimes_{i=0}^{n-1} \left( \sum_{k_i \in \{0,1\}}(-1)^{k_ia_i}\ketbras{k_i}\right)\\
        &= \sum_{\vec{k} \in \{0,1\}^n} (-1)^{\sum_i a_i k_i} \ketbras{\vec{k}} \\
        &= \sum_{\vec{k} \in \{0,1\}^n} (-1)^{\vec{k} \cdot \vec{a}} \ketbras{\vec{k}}, \label{eq:z_a_rep}
\end{align}
where $\vec{k} \cdot \vec{a}$ is the dot product of the binary vectors $\vec{k}$ and $\vec{a}$.

Next, we establish the explicit form of the operator product $X_{\vec{a}}Z_{\vec{b}}$ using the representations from \Cref{eq:x_a_rep} and \Cref{eq:z_a_rep}:
\begin{align}
    X_{\vec{a}}Z_{\vec{b}} &= \left(\sum_{\vec{k} \in \{0,1\}^n} \ketbra{\vec{k}}{\vec{k}\oplus \vec{a}}\right)\left(\sum_{\vec{l} \in \{0,1\}^n} (-1)^{\vec{l} \cdot \vec{b}} \ketbras{\vec{l}}\right)  \\
    &= \sum_{\vec{k} , \vec{l}\in \{0,1\}^n} (-1)^{\vec{l} \cdot \vec{b}}  \ketbra{\vec{k}}{\vec{k}\oplus \vec{a}}\ketbras{\vec{l}}  \\
    &= \sum_{\vec{k} , \vec{l}\in \{0,1\}^n} (-1)^{\vec{l} \cdot \vec{b}} \delta_{\vec{k}\oplus \vec{a}, \vec{l}} \ketbra{\vec{k}}{\vec{l}}\label{eq:XZ_delta}\\
    &=\sum_{\vec{k} \in \{0,1\}^n} (-1)^{(\vec{k} \oplus \vec{a}) \cdot \vec{b}}\ketbra{\vec{k}}{\vec{k}\oplus \vec{a}}. \label{eq:XZ_form_revised}
\end{align}
The transition to \Cref{eq:XZ_delta} introduces the Kronecker delta $\delta_{\vec{k}\oplus \vec{a}, \vec{l}}$ by evaluating the inner product $\braket{\vec{k}\oplus \vec{a}|\vec{l}}$ of the basis states (see \Cref{eq:def_kron_delta}).
Subsequently, to obtain the final expression, this Kronecker delta is used to eliminate the sum over $\vec{l}$ by enforcing the condition $\vec{l}=\vec{k} \oplus \vec{a}$, as all other terms are zero.
Using this derived representation of $X_{\vec{a}}Z_{\vec{b}}$, we can establish the identity presented in the following lemma, which is used in a subsequent proof.

\begin{lemma}\label{lemma_pauli_product}
    For $n$-bit binary vectors $\vec{a}, \vec{b}, \vec{c}, \vec{d} \in \{0,1\}^n$, the following identity holds:
    \begin{align}
        X_{\vec{a}}Z_{\vec{b}} X_{\vec{c}}Z_{\vec{d}} (X_{\vec{a}}Z_{\vec{b}})^{\dagger} =(-1)^{\vec{a}\cdot \vec{d}} (-1)^{\vec{c}\cdot \vec{b}} X_{\vec{c}}Z_{\vec{d}}
    \end{align}
\end{lemma}

\begin{proof}
Let $P_1=X_{\vec{a}}Z_{\vec{b}}$ and $P_2=X_{\vec{c}}Z_{\vec{d}}$.
Using the form from \Cref{eq:XZ_form_revised}:
\begin{align}
P_1 &= \sum_{\vec{k} \in \{0,1\}^n} (-1)^{(\vec{k} \oplus \vec{a}) \cdot \vec{b}} \ketbra{\vec{k}}{\vec{k}\oplus \vec{a}}, \\
P_2 &= \sum_{\vec{l} \in \{0,1\}^n} (-1)^{(\vec{l} \oplus \vec{c}) \cdot \vec{d}} \ketbra{\vec{l}}{\vec{l}\oplus \vec{c}}.
\end{align}
The Hermitian conjugate $P_1^{\dagger}$ is:
\begin{align}
P_1^\dagger = \left(\sum_{\vec{m} \in \{0,1\}^n} (-1)^{(\vec{m} \oplus \vec{a}) \cdot \vec{b}} \ketbra{\vec{m}}{\vec{m}\oplus \vec{a}}\right)^\dagger = \sum_{\vec{m} \in \{0,1\}^n} (-1)^{(\vec{m} \oplus \vec{a}) \cdot \vec{b}} \ketbra{\vec{m} \oplus \vec{a}}{\vec{m}}.
\end{align}

\noindent
First, we compute the product $P_1 P_2$:
\begin{align}
    P_1 P_2 &=\left(\sum_{\vec{k} \in \{0,1\}^n} (-1)^{(\vec{k} \oplus \vec{a}) \cdot \vec{b}}\ketbra{\vec{k}}{\vec{k}\oplus\vec{a}}\right)\left(\sum_{\vec{l} \in \{0,1\}^n} (-1)^{(\vec{l} \oplus \vec{c}) \cdot \vec{d}}\ketbra{\vec{l}}{\vec{l}\oplus\vec{c}}\right)\\
    &= \left(\sum_{\vec{k} , \vec{l}\in \{0,1\}^n} (-1)^{(\vec{k} \oplus \vec{a}) \cdot \vec{b}} (-1)^{(\vec{l} \oplus \vec{c}) \cdot \vec{d}}\ketbra{\vec{k}}{\vec{k}\oplus\vec{a}} \ketbra{\vec{l}}{\vec{l}\oplus\vec{c}}\right)\\
    &= \left(\sum_{\vec{k} , \vec{l}\in \{0,1\}^n} (-1)^{(\vec{k} \oplus \vec{a}) \cdot \vec{b}} (-1)^{(\vec{l} \oplus \vec{c}) \cdot \vec{d}}\delta_{\vec{k} \oplus \vec{a},\vec{l}}\ketbra{\vec{k}}{\vec{l}\oplus\vec{c}}\right) \label{eq:inner_prod_delta}\\
    &= \left(\sum_{\vec{k}\in \{0,1\}^n} (-1)^{(\vec{k} \oplus \vec{a}) \cdot \vec{b}} (-1)^{(\vec{k} \oplus \vec{a} \oplus \vec{c}) \cdot \vec{d}}\ketbra{\vec{k}}{\vec{k} \oplus \vec{a}\oplus\vec{c}}\right)\label{eq:canncel_sum_for_zero_terms}\\
    &= (-1)^{\vec{a}\cdot \vec{d}} \left(\sum_{\vec{k} \in \{0,1\}^n} (-1)^{(\vec{k} \oplus \vec{a}) \cdot \vec{b}} (-1)^{(\vec{k}\oplus \vec{c}) \cdot \vec{d}}\ketbra{\vec{k}}{\vec{k}\oplus \vec{a}\oplus \vec{c}}\right)
\end{align}
In \Cref{eq:inner_prod_delta}, we replaced the inner product $\bra{\vec{k}\oplus\vec{a}} \ket{\vec{l}}$ by the Kronecker delta $\delta_{\vec{k} \oplus \vec{a},\vec{l}}$, which is only non-zero for $\vec{l} = \vec{k}\oplus\vec{a}$.
By enforcing this condition, we omit summing over $\vec{l}$ in \Cref{eq:canncel_sum_for_zero_terms}.
Now, we multiply by $P_1^{\dagger}$:
\begin{align}
    P_1P_2P_1^\dagger &= (-1)^{\vec{a}\cdot \vec{d}} \left(\sum_{\vec{k} \{0,1\}^n} (-1)^{(\vec{k} \oplus \vec{a}) \cdot \vec{b}} (-1)^{(\vec{k}\oplus \vec{c}) \cdot \vec{d}}\ketbra{\vec{k}}{\vec{k}\oplus \vec{a}\oplus \vec{c}}\right)\left(\sum_{\vec{m} \in {0,1}^n} (-1)^{(\vec{m}\oplus\vec{a}) \cdot \vec{b}} \ketbra{\vec{m}\oplus \vec{a}}{\vec{m}}\right)\\
    &= (-1)^{\vec{a}\cdot \vec{d}}\sum_{\vec{k} , \vec{m} \{0,1\}^n} (-1)^{(\vec{k} \oplus \vec{a}) \cdot \vec{b}} (-1)^{(\vec{k}\oplus \vec{c}) \cdot \vec{d}}(-1)^{(\vec{m} \oplus \vec{a}) \cdot \vec{b}}\ketbra{\vec{k}}{\vec{k}\oplus \vec{a}\oplus \vec{c}}\ketbra{\vec{m}\oplus \vec{a}}{\vec{m}}.
\end{align}
By evaluating the inner product $\bra{\vec{k}\oplus \vec{a}\oplus \vec{c}}\ket{\vec{m}\oplus \vec{a}}$, we obtain only non-zero terms in the summation when $\vec{k}\oplus \vec{a}\oplus \vec{c}=\vec{m}\oplus \vec{a}$ (see \Cref{eq:def_kron_delta}), which is equivalent to $\vec{k}\oplus\vec{c}=\vec{m}$.
By directly enforcing this identity, the summation simplifies as follows:
\begin{align}
    P_1P_2P_1^\dagger &= (-1)^{\vec{a}\cdot \vec{d}}\sum_{\vec{k} \{0,1\}^n}(-1)^{(\vec{k} \oplus \vec{a}) \cdot \vec{b}} (-1)^{(\vec{k}\oplus \vec{c}) \cdot \vec{d}}(-1)^{(\vec{k} \oplus  \vec{c}  \oplus \vec{a}) \cdot \vec{b}}\ketbra{\vec{k}}{\vec{k}\oplus \vec{c}}\\
    &=(-1)^{\vec{a}\cdot \vec{d}} (-1)^{\vec{c}\cdot \vec{b}} \sum_{\vec{k} \{0,1\}^n}(-1)^{(\vec{k} \oplus \vec{a}) \cdot \vec{b}} (-1)^{(\vec{k} \oplus \vec{c}) \cdot \vec{d}}(-1)^{(\vec{k} \oplus \vec{a}) \cdot \vec{b}}\ketbra{\vec{k}}{\vec{k}\oplus \vec{c}}\\
    &=(-1)^{\vec{a}\cdot \vec{d}} (-1)^{\vec{c}\cdot \vec{b}} \sum_{\vec{k} \in \{0,1\}^n} (-1)^{(\vec{k} \oplus \vec{c}) \cdot \vec{d}}\ketbra{\vec{k}}{\vec{k}\oplus \vec{c}}\\
    &=(-1)^{\vec{a}\cdot \vec{d}} (-1)^{\vec{c}\cdot \vec{b}} X_{\vec{c}}Z_{\vec{d}}
\end{align}
    
\end{proof}

\subsection{Channel identities}

Building on the representations of $Z_{\vec{a}}$ operators, we now prove a useful identity relating a measure-and-prepare channel to an average over conjugation by $Z_{\vec{a}}$ operators.

\begin{lemma}\label{lemma_z_channel}
    For any operator $\rho$, it holds that
    \begin{align}
        \sum_{\vec{k}\in\{0,1\}^n}\tr\left[\ketbras{\vec{k}}\rho\right] \ketbras{\vec{k}} = \frac{1}{2^n}\sum_{\vec{a}\in\{0,1\}^n}Z_{\vec{a}} \rho Z_{\vec{a}}
    \end{align}
\end{lemma}
\begin{proof}
We start by using \Cref{eq:z_a_rep}:
    \begin{align}
        \frac{1}{2^n}\sum_{\vec{a}\in\{0,1\}^n}Z_{\vec{a}} \rho Z_{\vec{a}} &= \frac{1}{2^n}\sum_{\vec{a}\in\{0,1\}^n}\left(\sum_{\vec{k} \in \{0,1\}^n} (-1)^{\vec{k} \cdot \vec{a}} \ketbras{\vec{k}}\right)\rho \left(\sum_{\vec{l} \in \{0,1\}^n} (-1)^{\vec{l} \cdot \vec{a}} \ketbras{\vec{l}}\right)\\
        &= \frac{1}{2^n}\sum_{\vec{a},\vec{k},\vec{l}\in\{0,1\}^n} (-1)^{\vec{k} \cdot \vec{a}} (-1)^{\vec{l} \cdot \vec{a}}  \ketbras{\vec{k}}\rho \ketbras{\vec{l}}\\
        &= \frac{1}{2^n}\sum_{\vec{k},\vec{l}\in\{0,1\}^n}  \left(\ketbras{\vec{k}}\rho \ketbras{\vec{l}}\sum_{a\in\{0,1\}^n} (-1)^{\vec{k} \cdot \vec{a}} (-1)^{\vec{l} \cdot \vec{a}} \right)
    \end{align}
Next, we apply \Cref{lemma_double_minus_one}:
    \begin{align}
        \frac{1}{2^n}\sum_{\vec{a}\in\{0,1\}^n}Z_{\vec{a}} \rho Z_{\vec{a}}&= \frac{1}{2^n}\sum_{\vec{k},\vec{l}\in\{0,1\}^n}   \ketbras{\vec{k}}\rho \ketbras{\vec{l}}2^n\delta_{\vec{k},\vec{l}}\\
        &=\sum_{\vec{k}\in\{0,1\}^n}  \ketbras{\vec{k}}\rho \ketbras{\vec{k}} \label{eq:diag_form} \\
        &=\sum_{\vec{k}\in\{0,1\}^n}\tr\left[\ketbras{\vec{k}}\rho\right ]\ketbras{\vec{k}} \label{eq:diag_trace_form}.
    \end{align}
    The transition from \Cref{eq:diag_form} to \Cref{eq:diag_trace_form} uses the identity 
    \begin{align}
        \tr\left[\ketbras{\vec{k}}\rho\right] = \sum_{\vec{l}\in\{0,1\}^n} \bra{\vec{l}}\ketbras{\vec{k}}\rho\ket{\vec{l}} = \sum_{\vec{l}\in\{0,1\}^n}\delta_{\vec{l},\vec{k}} \braket{\vec{k}|\rho|\vec{l}} = \braket{\vec{k}|\rho|\vec{k}}.
    \end{align}
\end{proof}

Finally, we conclude with the following lemma, which demonstrates that averaging over the full Pauli group $\mathcal{Q}_n$ projects any operator onto the identity.
\begin{lemma}\label{lemma_pauli_twirl}
    For any operator $\rho$, it holds that
    \begin{align}
        \frac{1}{2^{2n}}\sum_{P\in \mathcal{Q}_n}P \rho P = \frac{\tr[\rho]}{2^n}I^{\otimes n}
    \end{align}
\end{lemma}
\begin{proof}
As $\mathcal{Q}_n$ forms a basis~\cite{Lawrence2002}, we can expand $\rho$ as
\begin{align}\label{eq:Pauli_expansion}
    \rho = \frac{1}{2^n} \sum_{\vec{a},\vec{b}\in \{0,1\}^n} x_{\vec{a},\vec{b}} X_{\vec{a}}Z_{\vec{b}}.
\end{align}
The coefficients $x_{\vec{a},\vec{b}}$ are isolated by taking the Hilbert-Schmidt inner product with a basis element $X_{\vec{a}}Z_{\vec{b}}$. 
For this representation, the general Pauli orthogonality relation from \Cref{eq:pauli_orthogonality} takes the specific form $\tr[(X_{\vec{a}}Z_{\vec{b}})(X_{\vec{c}}Z_{\vec{d}})]=2^n \delta_{\vec{a},\vec{c}}\delta_{\vec{b},\vec{d}}$. This property, together with the linearity of the trace, causes the sum to collapse:
\begin{align}
\tr[X_{\vec{a}}Z_{\vec{b}}\rho] &= \frac{1}{2^n} \sum_{\vec{c},\vec{d}} x_{\vec{c},\vec{d}} \tr[(X_{\vec{a}}Z_{\vec{b}})(X_{\vec{c}}Z_{\vec{d}})] \\
&= \frac{1}{2^n} \sum_{\vec{c},\vec{d}} x_{\vec{c},\vec{d}} (2^n \delta_{\vec{a},\vec{c}}\delta_{\vec{b},\vec{d}})\\ 
&= x_{\vec{a},\vec{b}}.\label{eq:x_a_b_coefficient}
\end{align}

In the following step, we substitute the Pauli expansion of $\rho$ from \Cref{eq:Pauli_expansion} and also expand the sum over $P=X_{\vec{a}}Z_{\vec{b}}\in\mathcal{Q}_n$. 
Crucially, we use the fact that all Pauli operators are Hermitian ($P=P^{\dagger}$) to write the final operator as its conjugate transpose:
\begin{align}
    \frac{1}{2^{2n}}\sum_{P\in \mathcal{Q}_n}P \rho P &= \frac{1}{2^{2n}}\sum_{P\in \mathcal{Q}_n}P \left(\frac{1}{2^n} \sum_{\vec{c},\vec{d}\in \{0,1\}^n} x_{\vec{c},\vec{d}} X_{\vec{c}}Z_{\vec{d}}\right) P\\
    &= \frac{1}{2^{3n}}\sum_{\vec{a}, \vec{c},\vec{b}, \vec{d}} x_{\vec{c},\vec{d}} X_{\vec{a}}Z_{\vec{b}} X_{\vec{c}}Z_{\vec{d}} (X_{\vec{a}}Z_{\vec{b}})^{\dagger}.
\end{align}
We now apply \Cref{lemma_pauli_product} and obtain:
\begin{align}
    \frac{1}{2^{2n}}\sum_{P\in \mathcal{Q}_n}P \rho P 
    &= \frac{1}{2^{3n}}\sum_{\vec{a}, \vec{b}, \vec{c}, \vec{d} \in \{0,1\}^n} x_{\vec{c},\vec{d}} (-1)^{\vec{a}\cdot \vec{d}} (-1)^{\vec{c}\cdot \vec{b}} X_{\vec{c}}Z_{\vec{d}}\label{eq:pauli_twirl_1}\\
    &= \frac{1}{2^{3n}}\sum_{\vec{c}, \vec{d} \in \{0,1\}^n} x_{\vec{c},\vec{d}}  X_{\vec{c}}Z_{\vec{d}}\left(\sum_{\vec{a}\in \{0,1\}^n} (-1)^{\vec{a}\cdot \vec{d}}\right)\left(\sum_{\vec{b} \in \{0,1\}^n} (-1)^{\vec{c}\cdot \vec{b}}\right).
\end{align}
Applying \Cref{lemma_minus_one} twice yields
\begin{align}
    \frac{1}{2^{2n}}\sum_{P\in \mathcal{Q}_n}P \rho P  &= \frac{1}{2^{3n}}\sum_{a', b'} x_{\vec{c},\vec{d}}  X_{\vec{c}}Z_{\vec{d}}2^n\delta_{\vec{d},\vec{0}}2^n\delta_{\vec{c},\vec{0}} \label{eq:pauli_twirl_2}\\
    &= \frac{1}{2^{n}} x_{\vec{0},\vec{0}} X_{\vec{0}}Z_{\vec{0}}\\
    &= \frac{\tr[\rho]}{2^n}I^{\otimes n}
\end{align}
where $x_{\vec{0},\vec{0}}=\tr[\rho]$ follows from \Cref{eq:x_a_b_coefficient}.
    
\end{proof}

 \section{Twirling a Pauli channel with a Pauli mixing ensemble}\label{sec:appendix_pauli_mixing}
This appendix provides a detailed proof of the following lemma, which establishes that twirling a Pauli channel with a Pauli mixing ensemble results in a depolarizing channel.

\begin{lemma} \label{lemma_pauli_mixing}
    Let $\mathcal{C}_{\text{Pauli}}$ be a Pauli channel and $\mathcal{E} = \{(p_i, U_i)\}_{i=0}^{K-1}$ a Pauli mixing unitary ensemble.
    Then it holds that
    \begin{align}
        \mathbb{E}_{\mathcal{E}}(\mathcal{C}_{\text{Pauli}}) = \mathcal{D}_{F(\mathcal{C}_{\text{Pauli}})}.
    \end{align}
\end{lemma}
\begin{proof}
We start by substituting the definitions of the $\mathcal{E}$-channel twirl from \Cref{eq:expectation_ensemble} and the Pauli channel $\mathcal{C}_{\text{Pauli}}$ from \Cref{eq:pauli_channel}:
    \begin{align}
        \mathbb{E}_{\mathcal{E}}(\mathcal{C}_{\text{Pauli}})(\rho)&= \sum_{i=0}^{K-1} p_iU_i^\dagger\mathcal{C}_{\text{Pauli}}(U_i\rho U_i^{\dagger})U_i\\
        &=  \sum_{i=0}^{K-1} p_iU_i^\dagger\left(\sum_{a=0}^{2^{2n}-1}\chi_{aa}P_aU_i\rho U_i^{\dagger} P_a\right)U_i \\
        &=  \sum_{a=0}^{2^{2n}-1} \chi_{aa} \sum_{i=0}^{K-1} p_i U_i^\dagger P_aU_i\rho U_i^{\dagger} P_a U_i.  \label{eq:proof_expanded_twirl}
\end{align}
Next, we separate the term for $a=0$, where $P_0=I^{\otimes n}$. 
For this term it holds that $U_i^\dagger P_0U_i = U_i^\dagger I^{\otimes n} U_i = I^{\otimes n}$.
Thus, the $a=0$ term of the sum in \Cref{eq:proof_expanded_twirl} is
\begin{align}
     \chi_{00}\sum_{i=0}^{K-1} p_i I^{\otimes n}\rho I^{\otimes n}=\chi_{00}\rho \left(\sum_{i=0}^{K-1} p_i \right) = \chi_{00}\rho.
\end{align}
So, \Cref{eq:proof_expanded_twirl} can be written as:
\begin{align}
    \mathbb{E}_{\mathcal{E}}(\mathcal{C}_{\text{Pauli}})(\rho) &=  \chi_{00}\rho + \sum_{a=1}^{2^{2n}-1} \chi_{aa} \sum_{i=0}^{K-1} p_i U_i^\dagger P_aU_i\rho U_i^{\dagger} P_a U_i.\label{eq:proof_P0_separated}
\end{align}

Now, we analyze the remaining terms in the sum for $a\ge 1$.
Since each $P_a\in \mathcal{Q}_n^*$ is Hermitian ($P_a^\dagger = P_a$), the resulting operator $U_i^{\dagger} P_a U_i$ is also Hermitian:
\begin{align}
    (U_i^{\dagger} P_a U_i)^{\dagger}= U_i^{\dagger} P_a^\dagger (U_i^{\dagger})^\dagger =U_i^{\dagger} P_a U_i.
\end{align}
From the definition of the Pauli mixing ensemble in \Cref{sec:twirl}, this operator $U_i^{\dagger} P_a U_i$ is in the Pauli group $\mathcal{P}_n$.
We can write it as  $U_i^{\dagger} P_a U_i = \eta_{ai}\pi(U_i^\dagger P_aU_i)$, where $\pi(U_i^\dagger P_aU_i)$ is the corresponding phase-free Pauli operator in $\mathcal{Q}_n^*$, and $\eta_{ai}\in \{\pm 1\}$ is a real sign factor since $U_i^{\dagger} P_a U_i$ is Hermitian.

Substituting $U_i^{\dagger} P_a U_i = \eta_{ai}\pi(U_i^\dagger P_aU_i)$ in \Cref{eq:proof_P0_separated}, we obtain
\begin{align}
    \mathbb{E}_{\mathcal{E}}(\mathcal{C}_{\text{Pauli}})(\rho) &= \chi_{00}\rho +  \sum_{a=1}^{2^{2n}-1} \chi_{aa} \sum_{i=0}^{K-1} p_i \eta_{ai}\pi(U_i^\dagger P_aU_i) \rho \eta_{ai}\pi(U_i^{\dagger} P_a U_i)\\
    &= \chi_{00}\rho +  \sum_{a=1}^{2^{2n}-1} \chi_{aa} \sum_{i=0}^{K-1} p_i \pi(U_i^\dagger P_aU_i) \rho \pi(U_i^{\dagger} P_a U_i) \label{eq:proof_with_pi},
\end{align}
where the sign factors cancel as $\eta_{ai}^2=1$.

To simplify the sum over the ensemble index $i$,  we use the Pauli mixing property of the ensemble $\mathcal{E}$. 
For clarity in the following derivation, we define a function $\Delta$ that acts as a Kronecker delta for Pauli operators~$P_a, P_b\in\mathcal{Q}_n$:
\begin{align}
    \Delta(P_a, P_b) = \begin{cases}
        1, & \text{if } P_a = P_b\\
        0, & \text{otherwise}.
    \end{cases}
\end{align}
According to the definition of a Pauli mixing ensemble from \Cref{eq:definition_pauli_mixing}, for any input operator $P_a\in \mathcal{Q}_n^*$, the random output operator $\pi(U_i^{\dagger} P_a U_i)$ (when $U_i$ is drawn from $\mathcal{E}$) is uniformly distributed over $\mathcal{Q}_n^*$. 
This means that for any target Pauli operator $P\in \mathcal{Q}_n^*$, it holds that
\begin{align}
 \sum_{\substack{i\\ \text{s.t. }\pi(U_i^{\dagger} P_a U_i)=P}} p_i &= \sum_{i=0}^{K-1} p_i\Delta\left(\pi(U_i^{\dagger} P_a U_i),P\right) \\ 
&= \frac{1}{|\mathcal{Q}_n^*|} \\
&= \frac{1}{2^{2n}-1}.
\end{align}
Therefore, for any $P_a\in \mathcal{Q}_n^*$, the inner sum in \Cref{eq:proof_with_pi}  can be simplified by regrouping the summation.
Instead of summing over the ensemble index $i$, we sum over all possible output Pauli operators $P\in \mathcal{Q}_n^*$ and collect the probabilities $p_i$ for each outcome.
This is achieved by inserting the identity $1=\sum_{P\in\mathcal{Q}_n^*}\Delta\left(\pi(U_i^{\dagger} P_a U_i),P\right)$ and swapping the summation order:
\begin{align}
    \sum_{i=0}^{K-1} p_i \pi(U_i^\dagger P_aU_i) \rho \pi(U_i^{\dagger} P_a U_i) &= \sum_{i=0}^{K-1} p_i \left(\sum_{P\in\mathcal{Q}_n^*}\Delta\left(\pi(U_i^{\dagger} P_a U_i),P\right)\right)\pi(U_i^\dagger P_aU_i) \rho \pi(U_i^{\dagger} P_a U_i) \\
    &= \sum_{P\in\mathcal{Q}_n^*} \left(\sum_{i=0}^{K-1} p_i\Delta\left(\pi(U_i^{\dagger} P_a U_i),P\right)\right)P \rho P \\
    &=\frac{1}{2^{2n}-1}\sum_{P \in \mathcal{Q}_n^*}P \rho P. \label{eq:proof_pauli_mixing_applied}
\end{align}

Substituting the result from \Cref{eq:proof_pauli_mixing_applied} into \Cref{eq:proof_with_pi} yields
\begin{align}
    \mathbb{E}_{\mathcal{E}}(\mathcal{C}_{\text{Pauli}})(\rho) &= \chi_{00}\rho +  \sum_{a=1}^{2^{2n}-1} \chi_{aa} \frac{1}{2^{2n}-1}\sum_{P \in \mathcal{Q}_n^*}P \rho P. 
\end{align}
Using the condition $\sum_{a=0}^{2^{2n}-1}\chi_{aa}=1$ of a Pauli channel, we have $\sum_{a=1}^{2^{2n}-1}\chi_{aa}=1-\chi_{00}$.
Thus:
\begin{align}
    \mathbb{E}_{\mathcal{E}}(\mathcal{C}_{\text{Pauli}})(\rho)&= \chi_{00}\rho +  (1-\chi_{00}) \frac{1}{2^{2n}-1}\sum_{P \in \mathcal{Q}_n^*}P \rho P.
\end{align}
This is the definition of a depolarizing channel $\mathcal{D}_{\chi_{00}}$.
Since $\chi_{00}=F(\mathcal{C}_{\text{Pauli}})$ (see \Cref{eq:ent_fidelity_chi}), we conclude:
\begin{align}
    \mathbb{E}_{\mathcal{E}}(\mathcal{C}_{\text{Pauli}}) = \mathcal{D}_{F(\mathcal{C}_{\text{Pauli}})}.
\end{align}
\end{proof} \section{Proof of \Cref{lemma_mp_channel}}\label{sec:proof_lemma}
In the following, \Cref{lemma_mp_channel} is restated for convenience.
\restatelemma*

\begin{proof}
We start by rewriting the  measure-and-prepare channel $\mathcal{M}$ as
    \begin{align}
        \mathcal{M}(\rho) &= \sum_{\vec{k}\in\{0,1\}^n}\tr\left[ \ketbras{\vec{k}}\rho\right]\rho_k \\
        &= \sum_{\vec{k}\in\{0,1\}^n}\tr\left[ \ketbras{\vec{k}}\rho\right]\sum_{\vec{l}\in\{0,1\}^n} \frac{1}{2^n-1}(1-\delta_{\vec{k},\vec{l}}) \ketbras{\vec{l}} \\
        &= \frac{1}{2^n-1}\sum_{\vec{k}\in\{0,1\}^n}\tr\left[ \ketbras{\vec{k}}\rho\right]\left(\left(\sum_{\vec{l}\in\{0,1\}^n}  \ketbras{\vec{l}}\right) - \ketbras{\vec{k}}\right) \\
        &= \frac{1}{2^n-1}\sum_{\vec{k}\in\{0,1\}^n}\tr\left[ \ketbras{\vec{k}}\rho\right]\left(I^{\otimes n}- \ketbras{\vec{k}}\right) \\
        &= \frac{1}{2^n-1}\left(\sum_{\vec{k}\in\{0,1\}^n}\tr\left[ \ketbras{\vec{k}}\rho\right]I^{\otimes n} -\sum_{\vec{k}\in\{0,1\}^n}\tr\left[ \ketbras{\vec{k}}\rho\right] \ketbras{\vec{k}}\right)\label{eq:mp_proof_1}.
\end{align}
Linearity of the trace shows that
\begin{align}
\sum_{\vec{k}\in\{0,1\}^n}\tr\left[\ketbras{\vec{k}}\rho\right] I^{\otimes n} = \tr\left[\sum_{\vec{k}\in\{0,1\}^n} \ketbras{\vec{k}}\rho\right]I^{\otimes n}= \tr\left[I^{\otimes n}\rho\right] I^{\otimes n}= \tr\left[\rho\right] I^{\otimes n}.
\end{align}
Furthermore, by  applying \Cref{lemma_z_channel} the second term becomes
\begin{align}
    \sum_{\vec{k}\in\{0,1\}^n}\tr\left[\ketbras{\vec{k}}\rho\right]\ketbras{\vec{k}} = \frac{1}{2^n}\sum_{\vec{a}\in\{0,1\}^n} Z_{\vec{a}} \rho Z_{\vec{a}}.
\end{align}
Thus, we can write \Cref{eq:mp_proof_1} as
\begin{align}
        \mathcal{M}(\rho) &= \frac{1}{2^n-1}\left(\tr\left[\rho\right]I^{\otimes n} - \frac{1}{2^n}\sum_{\vec{a}\in\{0,1\}^n}Z_{\vec{a}} \rho Z_{\vec{a}}\right)\label{eq:mp_proof_2}.
\end{align}
Next, we insert this result into the definition of the ensemble expectation value in \Cref{eq:expectation_ensemble} to express $\mathbb{E}_{\mathcal{V}}(\mathcal{M})$ as
    \begin{align}
        (\mathbb{E}_{\mathcal{V}}(\mathcal{M}))(\rho) &=\frac{1}{2^n+1} \sum_{j=0}^{2^n} V_j\mathcal{M}(V_j^\dagger \rho V_j)V_j^\dagger\\ 
        &= \frac{1}{2^n+1} \sum_{j=0}^{2^n} V_j\left(\frac{1}{2^n-1}\left(\tr\left[V_j^\dagger \rho V_j\right]I^{\otimes n} -\frac{1}{2^n}\sum_{\vec{a}\in\{0,1\}^n}Z_{\vec{a}} V_j^\dagger \rho V_j Z_{\vec{a}}\right) \right)V_j^\dagger\\
        &= \frac{1}{2^n+1} \sum_{j=0}^{2^n} V_j\left(\frac{1}{2^n-1}\left(\tr\left[\rho \right]I^{\otimes n} -\frac{1}{2^n}\sum_{\vec{a}\in\{0,1\}^n}Z_{\vec{a}} V_j^\dagger \rho V_j Z_{\vec{a}}\right) \right)V_j^\dagger\\
        &= \frac{1}{(2^n+1)(2^n-1)} \sum_{j=0}^{2^n} V_j\left(\tr\left[\rho \right]I^{\otimes n} -\frac{1}{2^n}\sum_{\vec{a}\in\{0,1\}^n}Z_{\vec{a}} V_j^\dagger \rho V_j Z_{\vec{a}}\right)V_j^\dagger\\
        &= \frac{1}{2^{2n}-1} \sum_{j=0}^{2^n}\left(\tr\left[\rho\right]I^{\otimes n} - \frac{1}{2^n}\sum_{\vec{a}\in\{0,1\}^n} V_jZ_{\vec{a}} V_j^\dagger \rho V_j Z_{\vec{a}} V_j^\dagger\right)\\
        &= \frac{1}{2^{2n}-1}\left((2^n+1)\tr\left[\rho\right]I^{\otimes n} - \frac{1}{2^n}\sum_{j=0}^{2^n}\sum_{\vec{a}\in\{0,1\}^n} V_jZ_{\vec{a}} V_j^\dagger \rho V_j Z_{\vec{a}} V_j^\dagger\right)\label{eq:mp_proof_3}.
    \end{align}
For the inner summation over $\vec{a}$, note that when $\vec{a}=\vec{0}$, it holds that $V_jZ_{\vec{a}} V_j^\dagger= V_jI^{\otimes n} V_j^\dagger = I^{\otimes n}$ for all $0\le j \le 2^n$.
Therefore, we can separate the identity term from the sum:
\begin{align}
    \sum_{\vec{a}\in\{0,1\}^n} V_jZ_{\vec{a}} V_j^\dagger \rho V_j Z_{\vec{a}} V_j^\dagger &= I^{\otimes n}\rho I^{\otimes n} + \sum_{\vec{a}\in\{0,1\}^n\setminus\{\vec{0}\}} V_jZ_{\vec{a}}V_j^{\dagger} \rho V_j Z_{\vec{a}} V_j^\dagger.\label{eq:separeted_identity}
\end{align}
For the remaining terms where $\vec{a}\neq \vec{0}$,  we use the definition of the maximal commuting set $S_j$ from \Cref{eq:S_j}. 
This definition establishes a one-to-one correspondence between the vectors $\vec{a}\in \{0,1\}^n\setminus\{\vec{0}\}$ and the Pauli operators $P\in S_j$ via the relation  $s_{j,\vec{a}}V_jZ_{\vec{a}}V^{\dagger}_j=P$, where  $s_{j,\vec{a}}\in\{-1,1\}$.
This allows us to rewrite the remaining sum by inserting the identity $1=s_{j,\vec{a}}^2$:
\begin{align}
    \sum_{\vec{a}\in\{0,1\}^n\setminus\{\vec{0}\}} V_jZ_{\vec{a}} V_j^\dagger \rho V_j Z_{\vec{a}} V_j^\dagger &=\sum_{\vec{a}\in\{0,1\}^n\setminus\{\vec{0}\}} s_{j,\vec{a}}^2 V_jZ_{\vec{a}}V_j^{\dagger} \rho V_j Z_{\vec{a}} V_j^\dagger\\
    &=  \sum_{\vec{a}\in\{0,1\}^n\setminus\{\vec{0}\}} s_{j,\vec{a}} V_jZ_{\vec{a}}V_j^{\dagger} \rho s_{j,\vec{a}} V_j Z_{\vec{a}} V_j^\dagger\\
    &=  \sum_{P \in S_j} P \rho P. \label{eq:sum_over_S_j}
\end{align}
Because the sets $\{S_j\}_{j=0}^{2^n}$ form a disjoint partition of $\mathcal{Q}_n^*$~\cite{Lawrence2002}, summing in \Cref{eq:mp_proof_3} over $j$ and substituting the expressions from \Cref{eq:separeted_identity,eq:sum_over_S_j} yields
\begin{align}
    (\mathbb{E}_{\mathcal{V}}(\mathcal{M}))(\rho) &= \frac{1}{2^{2n}-1}\left((2^n+1)\tr\left[\rho\right]I^{\otimes n} - \frac{2^n+1}{2^n}I^{\otimes n}\rho I^{\otimes n} - \frac{1}{2^n}\sum_{j=0}^{2^n}\sum_{P\in S_j}P \rho P\right)\\
    &= \frac{1}{2^{2n}-1}\left((2^n+1)\tr\left[\rho\right]I^{\otimes n} - \frac{2^n+1}{2^n}I^{\otimes n}\rho I^{\otimes n} - \frac{1}{2^n}\sum_{P\in \mathcal{Q}_n^*}P \rho P\right).
\end{align}
Furthermore, applying \Cref{lemma_pauli_twirl} to the first term yields
\begin{align}
    (\mathbb{E}_{\mathcal{V}}(\mathcal{M}))(\rho) &= \frac{1}{2^{2n}-1}\left(\frac{2^n+1}{2^{n}}\sum_{P\in \mathcal{Q}_n}P \rho P  - \frac{2^n+1}{2^n}I^{\otimes n}\rho I^{\otimes n} - \frac{1}{2^n}\sum_{P\in \mathcal{Q}_n^*}P \rho P\right)\\
    &= \frac{1}{2^{2n}-1}\left(\frac{2^n+1}{2^{n}}\sum_{P\in \mathcal{Q}_n^*}P \rho P  - \frac{1}{2^n}\sum_{P\in \mathcal{Q}_n^*}P \rho P\right)\label{eq:Q_star}\\
    &= \frac{1}{2^{2n}-1}\sum_{P\in \mathcal{Q}_n^*}P \rho P.\label{eq:final_equation_proof_lemma}
\end{align}
\Cref{eq:Q_star} is achieved by splitting the sum over $\mathcal{Q}_n$ into its identity component and the sum over the non-identity elements $\mathcal{Q}_n^*$, which causes the terms involving the identity operator to cancel. 
The final expression in \Cref{eq:final_equation_proof_lemma} matches the definition in \Cref{eq:depol_channel} of $\mathcal{D}_0$, thus $\mathcal{D}_0 = \mathbb{E}_{\mathcal{V}}(\mathcal{M})$.

To establish that the channel twirl $\mathbb{E}_{\mathcal{V}}(\mathcal{M})$ represents the minimal number of measure-and-prepare channels necessary to construct $\mathcal{D}_0$, we analyze the channel's \emph{Pauli transfer matrix}~(PTM)~\cite{Greenbaum2015}.
The PTM provides a $2^{2n}\times2^{2n}$ matrix representation of a quantum channel with $n$ input and output qubits by characterizing how the channel transforms each element of the Pauli basis $\mathcal{Q}_n$.

For the channel $\mathcal{D}_0$, the entries of the PTM $R$ are defined as:
\begin{align}
R_{ij} = \frac{1}{2^n}\tr[P_i\mathcal{D}_0(P_j)]
\end{align}
where $P_i, P_j \in \mathcal{Q}_n$ are Pauli operators indexed by $i, j$, and $P_0 = I^{\otimes n}$ by convention.
Each entry $R_{ij}$ therefore quantifies how the channel maps the input basis element $P_j$ to the output basis element $P_i$, representing the component of $P_i$ present in the transformed operator $\mathcal{D}_0(P_j)$.

Harada et al.~\cite[Theorem 1]{Harada2024} established that the minimal number $K$ of ancilla-free measure-and-prepare circuits required to construct an arbitrary channel as a (quasi)-probabilistic mixture satisfies
\begin{align}\label{eq:harada_inequality}
    \frac{\operatorname{Rank}(R) -1}{2^n -1} \le K.
\end{align}
For the depolarizing channel $\mathcal{D}_0$, the PTM $R$ is diagonal with entries~\cite{Greenbaum2015}:
\begin{align}
    R_{ij}= \begin{cases}
        1,  & \text{if } i=j=0\\
        \frac{1}{2^{2n}-1} & \text{if }  i=j\ne0
    \end{cases}
\end{align}
Since all diagonal entries are non-zero, $\operatorname{Rank}(R) = 2^{2n}$.
Substituting $\operatorname{Rank}(R) = 2^{2n}$ into Inequality~(\ref{eq:harada_inequality}) yields for the channel $\mathcal{D}_0$ that
\begin{align}
    K \geq \frac{2^{2n} - 1}{2^n - 1} = 2^n + 1.\label{eq:meas_and_prep_lb}
\end{align}
The channel twirl $\mathbb{E}_{\mathcal{V}}(\mathcal{M})$ provided in this work is a mixture of $K=2^n + 1$ measure-and-prepare circuit. 
Since this matches the lower bound in \Cref{eq:meas_and_prep_lb}, the number of circuits is minimal.

\end{proof} \section{Measuring entanglement fidelity}\label{sec:appendix_ent_fidelity}
To prove the method for measuring entanglement fidelity presented in \Cref{sec:comp_ent_fid}, we start by reformulating the depolarizing channel  $\mathcal{D}_0$ (\Cref{eq:depol_channel}) with zero entanglement fidelity:
\begin{align}
    \mathcal{D}_0(\rho) &= \frac{1}{2^{2n}-1}\sum_{P \in \mathcal{Q}_n^*} P \rho P\\
    &=\frac{1}{2^{2n}-1}\sum_{P \in \mathcal{Q}_n} P \rho P - \frac{1}{2^{2n}-1} \rho\\
    &= \frac{2^{n}}{2^{2n}-1} \tr[\rho]I^{\otimes n}  - \frac{1}{2^{2n}-1}  \rho,
\end{align}
where we applied \Cref{lemma_pauli_twirl} in the last equation.
For the probability of measuring the zero state when applying $\mathcal{D}_0$ to $\rho_0 = (\ketbras{0})^{\otimes n}$, we obtain
\begin{align}
    \braket{0^{\otimes n}|\mathcal{D}_0(\rho_0)|0^{\otimes n}}
    &= \frac{2^{n}}{2^{2n}-1} \underbrace{\braket{0^{\otimes n}|\tr[\rho_0]I^{\otimes n}|0^{\otimes n}}}_{=1}  - \frac{1}{2^{2n}-1} \underbrace{\braket{0^{\otimes n}|\rho_0|0^{\otimes n}}}_{=1} \\
    &= \frac{2^{n}}{2^{2n}-1} - \frac{1}{2^{2n}-1} \\
    &= \frac{2^{n} - 1}{2^{2n}-1} \\
    &= \frac{1}{2^{n}+1}.
\end{align}
Using the assumption from \Cref{theorem_arbitrary_cut} that we perform a $\mathcal{E}$-channel-twirl that transforms the channel $\mathcal{C}$ in a depolarizing channel $\mathcal{D}_{F(\mathcal{C})}$, i.e., $\mathbb{E}_{\mathcal{E}}(\mathcal{C})=\mathcal{D}_{F(\mathcal{C})}$, and  substituting the definition of the depolarizing channel $\mathcal{D}_{F(\mathcal{C})}$ from \Cref{eq:depol_channel_with_D0}, the probability $P_{0 \to 0}$ in \Cref{eq:P_0_to_0} becomes:
\begin{align}
    P_{0 \to 0} &= \braket{0^{\otimes n}|\mathbb{E}_{\mathcal{E}}(\mathcal{C})(\rho_0)|0^{\otimes n}}\\ 
    &= \braket{0^{\otimes n}|\mathcal{D}_{F(\mathcal{C})}(\rho_0)|0^{\otimes n}}\\ 
    &= F(\mathcal{C}) \braket{0^{\otimes n}|\mathcal{I}(\rho_0)|0^{\otimes n}} + (1 - F(\mathcal{C}))\braket{0^{\otimes n}|\mathcal{D}_0(\rho_0)|0^{\otimes n}}\\
    &= F(\mathcal{C}) + (1 - F(\mathcal{C}))\frac{1}{2^{n}+1} \\
    &= \left(1- \frac{1}{2^{n}+1}\right)F(\mathcal{C}) + \frac{1}{2^{n}+1}\\
    &= \frac{2^n}{2^{n}+1}F(\mathcal{C}) + \frac{1}{2^{n}+1}.
\end{align}
Isolating $F(\mathcal{C})$ by rearranging the terms yields the desired result:
\begin{align}
    F(C) = \frac{2^n+1}{2^n}P_{0 \to 0} - \frac{1}{2^n}.
\end{align} \section{Error}\label{sec:appendix_error}
The total estimation error $\epsilon(N)$ decomposes into sampling and systematic bias components through the triangle inequality:
\begin{align}
    \epsilon(N) &=  \left|\widehat{\braket{O}}_{\tilde{\mathcal{I}}(\rho)}^{N} - \tr[O\rho] \right|\\
    &\le \underbrace{\left|\widehat{\braket{O}}_{\tilde{\mathcal{I}}(\rho)}^{N}  - \tr[O\tilde{\mathcal{I}}(\rho)] \right|}_{=:\epsilon_{\text{sampling}}(N)} + \underbrace{\left| \tr[O\tilde{\mathcal{I}}(\rho)] - \tr[O\mathcal{I}(\rho)]\right|}_{=:\epsilon_{\text{bias}}} 
\end{align}
where we added $\pm\tr[O\tilde{\mathcal{I}}(\rho)]$. 
The operator $\tilde{\mathcal{I}}=\tilde{p}^{-1}\tilde{\mathcal{D}}_{p} - (\tilde{p}^{-1}-1)\tilde{\mathcal{D}}_{0}$ represents the actual implemented channel with the QPD using imperfect depolarizing channels $\tilde{\mathcal{D}}_{p}$ and $\tilde{\mathcal{D}}_{0}$.
The parameter $\tilde{p}$, an estimate of the ideal depolarization parameter $p$, represents the channel's estimated entanglement fidelity. 
As a fidelity, $\tilde{p}$ is naturally bounded above by $1$.
For the QPD to offer an advantage in sampling overhead over classical wire cutting, the true fidelity must satisfy $p>2^{-n}$ (see \Cref{eq:advantage_qpd_over_wire_cut}). Consequently, any useful estimate must also satisfy $\tilde{p}>2^{-n}$. 
This condition inherently ensures that $\tilde{p}$ is strictly positive, which is required for the inverse $\tilde{p}^{-1}$ in the QPD construction of $\tilde{\mathcal{I}}$ to be well-defined. 
Therefore, the bound for an estimate for an advantageous channel is $2^{-n}<\tilde{p}\le1$.

The sampling error $\epsilon_{\text{sampling}}$ can be bounded using  Hoeffding's inequality, as exemplified in related work~\cite[Proposition 3.1]{Piveteau2025}.
To ensure that the sampling error is within a desired precision $\epsilon_{\text{s}}$, i.e.,
\begin{equation} \label{eq:sampling_error_target_precision_revised}
\epsilon_{\text{sampling}}(N) \le \epsilon_{\text{s}},
\end{equation}
with a probability of at least $1-\delta$ (where $\delta \in (0,1]$), the required number of shots $N$ must satisfy~\cite{Piveteau2025}:
\begin{equation} \label{eq:N_shots_hoeffding_revised}
N \ge 2 \left(\frac{\kappa}{\epsilon_{\text{s}}}\right)^2\ln\left(\frac{2}{\delta}\right).
\end{equation}
In this bound, $\kappa=2\tilde{p}^{-1}-1$ is the sampling overhead of the implemented QPD. 
Rearanging Inequality~(\ref{eq:N_shots_hoeffding_revised}) yields that, for a fixed probability $\delta \in (0,1]$, the achievable sampling precision $\epsilon_{\text{s}}$ scales with the number of shots $N$ as $\epsilon_{\text{s}}\in\mathcal{O}(\kappa/\sqrt{N})$.

The systematic bias $\epsilon_{\text{bias}}$ can be bounded as
\begin{align}
     \epsilon_{\text{bias}} &= \left| \tr[O\tilde{\mathcal{I}}(\rho)] - \tr[O\rho] \right| \\
      &= \left| \tr\left[O \left(\tilde{\mathcal{I}}(\rho) - \rho\right)\right] \right| \\
      &\le \left\|O \left(\tilde{\mathcal{I}}(\rho) - \rho\right) \right\|_1.
 \end{align}
The final inequality follows from $|\tr[A]| = |\sum_{i}\lambda_i| \le \sum_{i}|\lambda_i| =\|A\|_1$ with  $\lambda_i$ representing the eigenvalues of~$A$ and $\|A\|_1$ being the trace norm of $A$.
Applying Hölder’s inequality $\|A^\dagger B\|_1 \le \|A\|_\infty\|B\|_1$ enables the norm of the observable~$O$ to be factored out~\cite{Larotonda2018}.
As $O=O^\dagger$, the bound becomes:
\begin{align}\label{eq:error_result_1}
     \epsilon_{\text{bias}} &\le \left\|O \right\|_{\infty} \left\|\tilde{\mathcal{I}}(\rho) - \rho \right\|_1.
 \end{align}
 
To eliminate the dependence on the state $\rho$, we employ the diamond norm, which upper bounds the trace norm for any input state:
\begin{align}\label{eq:error_result_2}
    \left\|\tilde{\mathcal{I}}(\rho) - \rho \right\|_1 = \left\|\tilde{\mathcal{I}}(\rho) - \mathcal{I}(\rho) \right\|_1 \le \left\|\tilde{\mathcal{I}} - \mathcal{I} \right\|_\diamond.
\end{align}
Decomposing the ideal identity $\mathcal{I}$ and implemented identity channel $\tilde{\mathcal{I}}$ using the QPD as $\mathcal{I}=p^{-1}\mathcal{D}_{p} - (p^{-1}-1)\mathcal{D}_{0}$  and $\tilde{\mathcal{I}}=\tilde{p}^{-1}\tilde{\mathcal{D}}_{p} - (\tilde{p}^{-1}-1)\tilde{\mathcal{D}}_{0}$, their difference becomes:
\begin{align}
    \left\|\tilde{\mathcal{I}} - \mathcal{I}\right\|_{\diamond} 
     &= \left\|\frac{1}{\tilde{p}}\tilde{\mathcal{D}}_p 
     -\left(\frac{1}{\tilde{p}}-1\right)\tilde{\mathcal{D}}_0 - \frac{1}{p}\mathcal{D}_p +\left(\frac{1}{p}-1\right)\mathcal{D}_0\right\|_{\diamond} 
\end{align}
By adding and subtracting terms $\frac{1}{\tilde{p}}\mathcal{D}_p$ and $(\frac{1}{\tilde{p}}-1)\mathcal{D}_0$, that weight ideal operators with the imperfect coefficients, we reorganize the expression to isolate the sources of errors:
\begin{align}
    \left\|\tilde{\mathcal{I}} - \mathcal{I}\right\|_{\diamond} 
     = \left\|\frac{1}{\tilde{p}}\left(\tilde{\mathcal{D}}_p - \mathcal{D}_p\right)
      - \left(\frac{1}{\tilde{p}}-1\right)\left(\tilde{\mathcal{D}}_0 -  \mathcal{D}_0\right) + \left(\frac{1}{\tilde{p}} - \frac{1}{p}\right)\mathcal{D}_p - \left(\frac{1}{\tilde{p}} - \frac{1}{p}\right)\mathcal{D}_0 \right\|_{\diamond}.
\end{align}
Applying the triangle inequality splits the right-hand side into three components:
\begin{align}
    \left\|\tilde{\mathcal{I}} - \mathcal{I}\right\|_{\diamond} 
     &\le \frac{1}{\tilde{p}}\left\|\left(\tilde{\mathcal{D}}_p - \mathcal{D}_p\right)\right\|_{\diamond} 
     +\left(\frac{1}{\tilde{p}}-1\right)\left\|\left(\tilde{\mathcal{D}}_0 - \mathcal{D}_0\right)\right\|_{\diamond} + \left|\frac{1}{\tilde{p}} - \frac{1}{p}\right|\left\|\mathcal{D}_p - \mathcal{D}_0\right\|_{\diamond}.\label{eq:coefficent_term_error}
\end{align}
The bound $2^{-n}<\tilde{p}\le1$, established previously, ensures that all coefficients in this expression are non-negative.
The depolarizing channel $\mathcal{D}_p$ can be formulated as $\mathcal{D}_p = p\mathcal{I} + (1-p)\mathcal{D}_0$, which simplifies the final term:
\begin{align}
    \left\|\mathcal{D}_p - \mathcal{D}_0 \right\|_{\diamond} &= \left\|p\mathcal{I} + (1-p)\mathcal{D}_0 - \mathcal{D}_0 \right\|_{\diamond}\\ 
     &= \left\|p\mathcal{I} - p\mathcal{D}_0 \right\|_{\diamond}\\ 
    &= p\left\|\mathcal{I} - \mathcal{D}_0 \right\|_{\diamond}\\
&\le 2p
\end{align}
where the last equality holds since the maximal value for the diamond norm distance is $2$ (see \Cref{eq:bounds_diamond_norm}).
As a result, we can express the last term of Inequality~(\ref{eq:coefficent_term_error}) as
\begin{align}\label{eq:error_result_4}
     \left|\frac{1}{\tilde{p}} - \frac{1}{p}\right|\left\|\mathcal{D}_p - \mathcal{D}_0\right\|_{\diamond} &\le 2p \left|\frac{1}{\tilde{p}} - \frac{1}{p}\right|
     = 2 \left|\frac{p}{\tilde{p}} - 1\right|.
\end{align}
Combining the results of \Cref{eq:error_result_1,eq:error_result_2,eq:coefficent_term_error,eq:error_result_4} yields the composite bias bound:
\begin{align}
     \epsilon_{\text{bias}} &\le \left\|O \right\|_{\infty} \left(\frac{1}{\tilde{p}}\left\|\tilde{\mathcal{D}}_p -\mathcal{D}_p \right\|_{\diamond} + \left(\frac{1}{\tilde{p}}-1\right)\left\|\mathcal{D}_0 - \tilde{\mathcal{D}}_0 \right\|_{\diamond} +  2 \left|\frac{p}{\tilde{p}} - 1\right| \right).
 \end{align}
which explicitly quantifies contributions from imperfections in the implemented depolarizing channels $\mathcal{D}_p$ (first summand) and $\mathcal{D}_0$ (second summand) and discrepancies between the ideal QPD coefficient $p^{-1}$ and its approximated value $\tilde{p}^{-1}$ (third summand).

To derive \Cref{eq:pauli_mixing_coherent_errors}, which quantifies the error when applying Pauli mixing ensembles $\mathcal{E}$ to non-Pauli channels $\mathcal{C}$, we first decompose the channel $\mathcal{C}$ into its diagonal (Pauli) and off-diagonal (coherent) parts based on its $\chi$-matrix representation from \Cref{eq:coherent_incoherent_part}:
\begin{align}
\mathcal{C}(\rho) = \mathcal{C}_{\text{Pauli}}(\rho) + \mathcal{C}_{\text{coherent}}(\rho),
\end{align}
where $\mathcal{C}_{\text{Pauli}}(\rho) = \sum_{i}\chi_{ii}P_i\rho P_i$ and $\mathcal{C}_{\text{coherent}}(\rho) = \sum_{i\ne j}\chi_{ij}P_i\rho P_j$.
Due to the linearity of the twirling superoperator $\mathbb{E}_{\mathcal{E}}$, the twirled channel can be similarly decomposed:
\begin{align}
    \mathbb{E}_{\mathcal{E}}(\mathcal{C}) &= \mathbb{E}_{\mathcal{E}}(\mathcal{C}_{\text{Pauli}}) +  \mathbb{E}_{\mathcal{E}}(\mathcal{C}_{\text{coherent}}).
\end{align}
From \Cref{lemma_pauli_mixing}, we know that twirling the Pauli part yields a depolarizing channel $\mathbb{E}_{\mathcal{E}}(\mathcal{C}_{\text{Pauli}}) =  \mathcal{D}_{F(\mathcal{C_{\text{Pauli}}})}$.
Furthermore, the entanglement fidelity of a channel depends only on the $\chi_{00}$ element of the $\chi$-matrix, so $F(\mathcal{C}) = F(\mathcal{C}_{\text{Pauli}}) = \chi_{00}$.
Substituting this gives
\begin{align}
    \mathbb{E}_{\mathcal{E}}(\mathcal{C}) &=  \mathcal{D}_{F(\mathcal{C})} +  \mathbb{E}_{\mathcal{E}}(\mathcal{C}_{\text{coherent}}).
\end{align}

With this expression, we can now bound the deviation from the ideal depolarizing channel. 
The derivation proceeds as follows:
\begin{align}
    \left\|\mathbb{E}_{\mathcal{E}}(\mathcal{C}) -\mathcal{D}_{F(\mathcal{C})}  \right\|_{\diamond} &= \left\|\mathcal{D}_{F(\mathcal{C})} + \mathbb{E}_{\mathcal{E}}(\mathcal{C}_{\text{coherent}}) -\mathcal{D}_{F(\mathcal{C})}  \right\|_{\diamond}\\
    &= \left\|\mathbb{E}_{\mathcal{E}}(\mathcal{C}_{\text{coherent}})\right\|_{\diamond}\\
    &= \max_{\rho\in D(R\otimes A)} \left\|(\mathcal{I}_R  \otimes\mathbb{E}_{\mathcal{E}}(\mathcal{C}_{\text{coherent}}) )(\rho)\right\|_1, \label{eq:choherent_error_trace_norm}
\end{align}
where we applied the definition of the diamond norm from \Cref{eq:diamond_norm_def}, which involves maximizing over all states $\rho$ on the combined system $R\otimes A$, where $A$ is the input Hilbert space of the channel $\mathcal{C}$ and $R$ is an ancillary system of $n_R$ qubits. 
Next, we substitute the definitions of the $\mathcal{E}$-channel-twirl (\Cref{eq:expectation_ensemble}) and the coherent error term $\mathcal{C}_{\text{coherent}}(\rho) = \sum_{i\ne j}\chi_{ij}P_i\rho P_j$:
\begin{align}
    \left\|\mathbb{E}_{\mathcal{E}}(\mathcal{C}) -\mathcal{D}_{F(\mathcal{C})}  \right\|_{\diamond}&= \max_{\rho\in D(R\otimes A)} \left\|\sum_{k=0}^{K-1} \sum_{i\ne j}p_k\chi_{ij}(I^{\otimes n_R} \otimes U_k^\dagger P_i U_k )\rho (I^{\otimes n_R} \otimes U_k^{\dagger} P_jU_k)\right\|_1\\
    &\le  \sum_{k=0}^{K-1} \sum_{i\ne j}p_k|\chi_{ij}|\max_{\rho\in D(R\otimes A)} \left\|(I^{\otimes n_R} \otimes U_k^\dagger P_i U_k)\rho (I^{\otimes n_R} \otimes U_k^{\dagger} P_jU_k)\right\|_1\\
    &=  \sum_{k=0}^{K-1} \sum_{i\ne j}p_k|\chi_{ij}|\max_{\rho\in D(R\otimes A)} \left\|\rho \right\|_1 \label{eq:proof_error_chi_1}\\
    &=  \sum_{k=0}^{K-1} \sum_{i\ne j}p_k|\chi_{ij}| \label{eq:proof_error_chi_2}\\
    &= \sum_{i\ne j}|\chi_{ij}| \label{eq:proof_error_chi_3}
\end{align}
\Cref{eq:proof_error_chi_1} follows directly from the unitary invariance of the trace norm, which guarantees invariance under unitary transformations of the form $I^{\otimes n_R} \otimes U_k^{\dagger} P_jU_k$ applied to both sides of $\rho$.
In \Cref{eq:proof_error_chi_2}, the trace norm $\|\rho\|_1$ simplifies to $1$ because $\rho$ is a valid density operator (positive semi-definite with $\tr[\rho]=1$), eliminating the need of the maximization over $\rho$.
Finally, \Cref{eq:proof_error_chi_3} uses the normalization condition $\sum_k p_k = 1$ for the probability distribution ${p_k}$. 

\end{document}